\documentclass[11pt, a4paper, reqno]{article}
\usepackage{amsmath,amsthm,amssymb}
\usepackage[english]{babel}
\usepackage{a4wide}
\usepackage{microtype}
\usepackage{hyperref}
\hypersetup{colorlinks=true,citecolor=blue,linkcolor=blue,urlcolor=blue}
\usepackage{url}
\usepackage{mathtools}
\usepackage{tikz}

\newtheorem{theorem}{Theorem}
\newtheorem*{theorem*}{Theorem}
\newtheorem{lemma}[theorem]{Lemma}
\newtheorem{corollary}[theorem]{Corollary}
\newtheorem{proposition}[theorem]{Proposition}
\newtheorem*{proposition*}{Proposition}
\newtheorem{claim}{Claim}
\newtheorem{observation}[theorem]{Observation}
\newtheorem{definition}[theorem]{Definition}
\theoremstyle{definition}
\newtheorem{remark}[theorem]{Remark}
\newtheorem{example}{Example}

\newcommand{\edges}[1]{E(#1)}
\newcommand{\vertices}[1]{V(#1)}
\newcommand{\points}[1]{P(#1)}

\newcommand{\pointgraph}[1]{pg(#1)}
\newcommand{\mim}[1]{\textnormal{MIM}(#1)}
\newcommand{\inc}[1]{inc(#1)}
\newcommand{\lat}[1]{S(#1)}
\newcommand{\T}{{\ensuremath{\mathcal T}}}

\newcommand{\covernum}[1]{\textnormal{cn}(#1)}
\newcommand{\bcovernum}[1]{\beta\textnormal{-cn}(#1)}
\newcommand{\fraccovernum}[1]{\textnormal{fcn}(#1)}
\newcommand{\cover}[1]{\textnormal{cn}(#1)}
\newcommand{\coverfrac}[1]{\textnormal{fcn}(#1)}
\newcommand{\vc}[1]{\textnormal{VC}(#1)}
\newcommand{\is}[1]{\textnormal{IS}(#1)}
\newcommand{\fhw}[1]{\textnormal{fhw}(#1)}
\newcommand{\bfhw}[1]{\beta\textnormal{-fhw}(#1)}
\newcommand{\hw}[1]{\textnormal{hw}(#1)}
\newcommand{\bhw}[1]{\beta\textnormal{-hw}(#1)}
\newcommand{\bsw}[1]{\beta\textnormal{-subw}(#1)}
\newcommand{\subw}[1]{\textnormal{subw}(#1)}
\newcommand{\tw}[1]{\textnormal{tw}(#1)}

\newcommand{\defeq}{\vcentcolon=}

\newcommand{\pw}[1]{\textnormal{pw}(#1)}
\newcommand{\fpw}[1]{\textnormal{fpw}(#1)}
\newcommand{\spw}[1]{\textnormal{spw}(#1)}
\newcommand{\mimw}[1]{\textnormal{mimw}(#1)}
\newcommand{\cw}[1]{\textnormal{cw}(#1)}
\newcommand{\ar}[1]{\textnormal{ar}(#1)}
\newcommand{\tuple}[1]{\mathbf{#1}}
\newcommand{\tvalue}[1]{\textnormal{val}(#1)}
\newcommand{\tvaluealg}[1]{\textnormal{val}_{\text{alg}}(#1)}
\newcommand{\opt}[1]{\textnormal{opt}(#1)}
\newcommand{\psiopt}{\psi_{\textnormal{opt}}}
\newcommand{\problem}[1]{{#1}}
\def\qplus{\mathbb{Q}_{\geq 0}}

\begin{document}

\title{Point-width and Max-CSPs\thanks{An extended abstract of this work appeared in the
\emph{Proceedings of the 34th Annual ACM/IEEE Symposium on Logic in Computer
Science (LICS'19)}~\cite{crz19:lics}. Stanislav
\v{Z}ivn\'y was supported by a Royal Society University Research Fellowship.
This project has received funding from the European Research Council (ERC) under
the European Union's Horizon 2020 research and innovation programme (grant
agreement No 714532). The paper reflects only the authors' views and not the
views of the ERC or the European Commission. The European Union is not liable
for any use that may be made of the information contained therein. Work done
while Cl\'ement Carbonnel and Miguel Romero were at the University of Oxford.}}

\author{
Cl\'ement Carbonnel\\
CNRS, University of Montpellier, France\\
\texttt{clement.carbonnel@lirmm.fr}
\and
Miguel Romero\\
Facultad de Ingenier\'ia y Ciencias, Universidad Adolfo Ib\'a\~nez\\
Santiago, Chile\\
\texttt{miguel.romero.o@uai.cl}
\and
Stanislav \v{Z}ivn\'{y}\\
University of Oxford, UK\\
\texttt{standa.zivny@cs.ox.ac.uk}
}

\date{}
\maketitle

\begin{abstract}

The complexity of (unbounded-arity) Max-CSPs under structural restrictions is poorly understood. The two most general hypergraph properties known to ensure tractability of Max-CSPs, $\beta$-acyclicity and bounded (incidence) MIM-width, are incomparable and lead to very different algorithms. 

We introduce the framework of point decompositions for hypergraphs and use it to derive a new sufficient condition for the tractability of (structurally restricted) Max-CSPs, which generalises both bounded MIM-width and $\beta$-acyclicity. On the way, we give a new characterisation of bounded MIM-width and discuss other hypergraph properties which are relevant to the complexity of Max-CSPs, such as $\beta$-hypertreewidth.

\end{abstract}

\section{Introduction}

The \emph{Constraint Satisfaction Problem} (\problem{CSP}) is a well-known
framework for expressing a wide range of both theoretical and real-life
combinatorial problems~\cite{montanari74:constraints,Jeavons98:algebraic,feder98:monotone}. Some
examples are satisfiability~\cite{Schaefer78:complexity}, evaluation of
conjunctive queries~\cite{chandra77:stoc,kolaitis98:pods}, graph
colourings~\cite{Hell90:h-coloring} and homomorphisms~\cite{hell:graphs}. An
instance of the \problem{CSP} is a set of \emph{variables}, a domain of
\emph{values} and a set of \emph{constraints}; each constraint is a relation
applied to a subset of the variables called the constraint scope. Given a
\problem{CSP} instance, the goal is to decide whether one can assign a value to
each variable so that all constraints are satisfied; that is, whether for every
constraint, the assignment restricted to the constraint scope belongs to the
constraint relation. Due to its expressivity, it is not surprising that the
\problem{CSP} is NP-complete in general. This has motivated a long line of research aiming
to find tractable restrictions of the problem, sometimes called \emph{islands of
tractability}. The focus of this paper is on the so-called \emph{structural
restrictions}, which restricts the ways in which the constraints overlap and
intersect each other.

A standard way of analysing structural restrictions is via the \emph{underlying hypergraph} of a \problem{CSP} instance. 
The vertex set of this hypergraph is the set of variables $X$ of the instance and the edges correspond to the scopes of the constraints: 
each constraint whose scope is a subset $S\subseteq X$ yields the edge $S$.  
Given a class $\mathcal{H}$ of hypergraphs, we define the problem
\problem{CSP($\mathcal{H},-$)} as the restriction of the \problem{CSP} to instances 
whose underlying hypergraphs lie in $\mathcal{H}$. Then the goal is to understand for
which classes $\mathcal{H}$ the problem \problem{CSP($\mathcal{H},-$)} is tractable, and 
for which classes $\mathcal{H}$ it is not.

The situation of CSP instances of \emph{bounded arity} (i.e., the maximum edge
size in the class $\mathcal{H}$ is a constant) is by now well-understood. 
In this setting, it follows from~\cite{freuder90:complexity}
and~\cite{Grohe07:jacm} (see also~\cite{GSS01}) that
\problem{CSP($\mathcal{H},-$)} is tractable if and only if $\mathcal{H}$ has \emph{bounded treewidth} 
(under the complexity theoretical assumption that FPT $\ne$ W[1]). 
On the other hand, the case of \emph{unbounded arity}, that is, arbitrary
classes $\mathcal{H}$ of hypergraphs, is more delicate. 
Unlike the bounded-arity case, the complexity of the problem heavily depends on how the constraints in a CSP instance are represented~\cite{chen10:succinct}. 
We focus on one of the most natural and well-studied representation of constraints, namely the \emph{positive} representation, 
where each constraint is represented by the list of tuples satisfying the constraint. 

Bounded treewidth is not the right answer for tractability in the case of unbounded arity, as
one can easily find classes $\mathcal{H}$ of hypergraphs of unbounded treewidth
such that \problem{CSP($\mathcal{H},-$)} is tractable. 
One of the first such classes are the \emph{acyclic} hypergraphs~\cite{BFMMUY81,Beeri83:jacm,Yannakakis81:vldb} (also called \emph{$\alpha$-acyclic}~\cite{Fagin83:jacm-degrees}). 
This tractability result has been extended to more general classes such as hypergraphs of bounded \emph{hypertreewidth}~\cite{Gottlob02:jcss-hypertree} 
and bounded \emph{fractional hypertreewidth}~\cite{Grohe14:talg}. The latter is the most general natural hypergraph property known to be tractable, although the precise borderline of polynomial-time solvability is still unknown (and cannot coincide with bounded fractional hypertreewidth; see~\cite{Marx13:jacm} for a brief discussion on that topic). However, as shown in~\cite{Marx13:jacm}, 
the classes $\mathcal{H}$ for which \problem{CSP($\mathcal{H},-$)} becomes \emph{fixed-parameter tractable} (parameterised by the size of the hypergraph) 
are precisely those of bounded \emph{submodular width}, which are 
more general than classes of hypergraphs of bounded fractional hypertreewidth.

In this paper we study the problem \problem{Max-CSP}\footnote{A usual definition
of a Max-CSP instance is a CSP instance with the goal to maximise the number of satisfied
constraints. As we explain in Section~\ref{sec:max-csp}, we actually consider a
more general framework, sometimes called \emph{finite-valued}
CSPs~\cite{tz16:jacm} or Max-CSPs with payoff
functions~\cite{Raghavendra08:everycsp}. Since our main result is a tractability
result, this makes it only stronger.}, which is a well-known generalisation of \problem{CSPs} for expressing optimisation problems. Now each constraint is of the form $f(\tuple{x})$, where $|\tuple{x}|=r$ and $f$ is an $r$-ary (finite-valued) function $f:D^r\to\qplus$ (we assume that $f$ is given as the set of pairs $\{(\tuple{d},f(\tuple{d})): \tuple{d}\in D^r, f(\tuple{d})>0\}$, which corresponds to the positive representation). 
Given a set of variables $X=\{x_1,\dots,x_n\}$, a domain $D$ of values and a set
$\mathcal{C}$ of (finite-valued) constraints, the goal is to compute the maximum value of 
$f(x_1,\dots,x_n)=\sum_{f_c(\tuple{x})\in \mathcal{C}} f_c(\tuple{x})$, over all possible assignments of values to $X$. 

In the case of bounded arity, tractability of \problem{Max-CSP($\mathcal{H},-$)} is also characterised by bounded treewidth, 
which follows directly from the CSP case. However, the complexity of unbounded-arity Max-CSPs under structural restrictions is poorly understood and the 
techniques used in the CSP context cannot be easily applied. 
Indeed, \problem{Max-CSP($\mathcal{H},-$)} is hard even for classes $\mathcal{H}$ of $\alpha$-acyclic hypergraphs~\cite{Gottlob09:icalp}. 
Moreover, unlike the CSP case, there is no \emph{known} maximal hypergraph property that leads to tractability. 
The two most general hypergraph properties known to ensure tractability of
\problem{Max-CSP($\mathcal{H},-$)} are \emph{$\beta$-acyclicity}\footnote{In fact, the authors in \cite{brault-baron15:stacs} consider a more general framework called the CSP \emph{with default values}, and
focus on counting solutions. However, they briefly discuss how to adapt the
results to the maximisation version.}~\cite{brault-baron15:stacs}, introduced
in~\cite{Fagin83:jacm-degrees}, and having \emph{bounded (incidence)
MIM-width}\footnote{The results for MIM-width in~\cite{Vatshelle12:phd,Saether15:jair} apply
to Max-SAT (and \#SAT), but can be adapted to Max-CSPs. Let us also remark that in~\cite{Vatshelle12:phd,Saether15:jair} a more general notion than that of bounded
MIM-width, namely having polynomial
\emph{PS-width}, is shown to be tractable for Max-SAT and \#SAT. This notion is however not purely structural, as it depends on the entire input instance and not just its hypergraph.}~\cite{Vatshelle12:phd,Saether15:jair}. 
These properties are incomparable~\cite{brault-baron15:stacs} and lead to very different algorithms. 
The main goal of this paper is to provide a common explanation for these two tractable properties, 
and in particular, for all known tractable hypergraph properties for
\problem{Max-CSPs}. 
We believe that such a unified explanation is a necessary first step to a better
understanding of the tractable structural restrictions of \problem{Max-CSPs}, 
and ultimately, to a precise characterisation of the tractability frontier.

\subsection{Contributions}

As our main contribution, we introduce the notions of \emph{point decomposition} and \emph{point-width}
that unify $\beta$-acyclicity and bounded MIM-width. 
We show that Max-CSPs (with positive representation) are tractable for 
hypergraphs of bounded point-width, 
provided a point decomposition of polynomial size and bounded width is also part of the input (Theorem~\ref{thm:alg}). 
Our tractability result explains the tractability of $\beta$-acyclic and bounded MIM-width hypergraphs. 
In particular, we prove that every $\beta$-acyclic hypergraph has a point decomposition of width $1$ and polynomial size (Theorem~\ref{thm:beta-acyc}), which can be computed in polynomial time.
In the case of MIM-width, we obtain a stronger result that may be of independent interest: 
having bounded MIM-width is equivalent to having bounded \emph{flat point-width} (Theorem~\ref{theo:mim-fpw}), 
where the latter is defined via a syntactic restriction of point decompositions. 
Finally, we also discuss some related notions such as
$\beta$-hypertreewidth~\cite{GP04} (Section~\ref{sec:conc}). 

The high-level idea behind our new notion of width is that a point decomposition of width $k\geq 1$ for a hypergraph $H$ provides a mechanism 
to encode \emph{several} tree decompositions of hypertreewidth at most $k$ in a compact and controlled way. In particular, 
a point decomposition will be expressive enough to encode one such a tree decomposition for \emph{each} subhypergraph of $H$. 
Interestingly, the underlying trees of all these tree decompositions can be very different from each other, as long as they respect the ``template'' tree $T$ 
given by the point decomposition. 
For \emph{flat} point decompositions, which capture MIM-width, these underlying trees need to be \emph{subtrees} of the template $T$, and then 
they are more similar to each other. 
The full details of point decompositions and their flat variant are given in
Sections~\ref{sec:pointdec} and~\ref{sec:mim}, respectively.

The algorithm behind our main tractability result (Theorem~\ref{thm:alg}) uses a form of dynamic programming over the point decomposition where in each step 
we need to solve an instance of the \emph{weighted maximum independent set}
problem in \emph{chordal graphs} (which is known to be tractable and in fact
solvable in linear time~\cite{frank1975some}, see also~\cite{Tarjan85:dm}). 
We can think of this procedure as doing dynamic programming \emph{simultaneously} over all the tree decompositions of the subhypergraphs of $H$ 
encoded in the point decomposition. 

\subsection{Related work}

It is also possible to parameterise CSPs and Max-CSPs by a class of admissible underlying \emph{structures}, instead of hypergraphs, 
which offers a more fine-grained analysis. 
In the case of CSPs of bounded arity, 
a complete classification of the tractable cases in terms of the underlying relational structures follows from~\cite{Dalmau02:width} and~\cite{Grohe07:jacm}. 
Recently, a similar classification has been obtained for (finite-valued)
Max-CSPs in terms of the underlying \emph{valued} structures~\cite{CRZ18}. 

Another important type of restrictions (and perhaps the most studied one) are the \emph{non-uniform} restrictions, where 
the constraint relations (or functions) are restricted to be fixed. 
In this case, the situation is fairly clear and now, after two decades of
intense research, complete classifications have
been obtained for CSPs~\cite{Bulatov17:focs,Zhuk17:focs}, 
and (finite-valued) Max-CSPs~\cite{tz16:jacm}.

\subsection{Structure} 
The paper is organised as follows.
Section~\ref{sec:prelim} introduces
the necessary notation on hypergraphs and Max-CSPs. Section~\ref{sec:pointdec}
defines point decompositions and point-width. The main tractability result is given in
Section~\ref{sec:alg}. Sections~\ref{sec:beta} and~\ref{sec:mim} show that $\beta$-acyclicity
and bounded MIM-width are special cases of bounded point-width, respectively. We
conclude in Section~\ref{sec:conc}.

\section{Preliminaries}
\label{sec:prelim}

\subsection{Hypergraphs, points and covers}

We assume that the reader is familiar with elementary graph
theory and refer to Diestel's textbook for more details~\cite{Diestel10:graph}. 
Given a graph $G$, we use $V(G)$ and $E(G)$ to denote its sets of \emph{vertices} and \emph{edges}, respectively. 
The subgraph of a graph $G$ \emph{induced} by a set $X\subseteq V(G)$, denoted by $G[X]$, has 
vertex set $X$ and edge set $\{\{u,v\}\in E(G): u,v\in X\}$. We use the same notation for directed graphs. 

\medskip
\noindent
{\bf Hypergraphs.\,} 
A (finite) hypergraph is a finite set of non-empty finite sets called
\emph{edges}. The set of \emph{vertices} of a hypergraph $H$, denoted by
$V(H)$, is the union of all its edges. Note that in this definition,
every vertex of a hypergraph belongs to at least one edge. A
\emph{subhypergraph} of a hypergraph $H$ is a subset of $H$. We use $\lat{H}$ to
denote the set of all vertex sets of subhypergraphs of $H$. 

\medskip
\noindent
{\bf Points.\,}
A \emph{point} of a hypergraph $H$ is a pair $(v,e)$ with $e \in H$ and $v \in
e$. We use $\points{H}$ to denote the set of all points of $H$. Given $P
\subseteq \points{H}$ and $e \in H$, the \emph{restriction of $e$ to $P$},
denoted by $e|_{P}$, is the set $\{v \in e : (v,e) \in P\}$. By extension the
\emph{restriction of $H$ to $P$}, denoted by $H|_{P}$, is the hypergraph $\{
  e|_{P} : e \in H, \, e|_{P} \neq \emptyset\}$. If $H'$ is a subhypergraph of
$H$ and $P \subseteq P(H)$, we use the notation $H'|_P$ as a shorthand for
$H'|_{P \cap \points{H'}}$.

\medskip
\noindent
{\bf Covers.\,}
An \emph{edge cover} of a hypergraph $H$ is a subhypergraph $C$ of $H$ such that
$\vertices{C} = \vertices{H}$. The \emph{cover number} of $H$, denoted by
$\covernum{H}$, is the smallest cardinality of an edge cover of $H$. We denote
by $\bcovernum{H}$ the maximum of $\covernum{H'}$ over all subhypergraphs $H'$
of $H$. 

\subsection{Max-CSP}
\label{sec:max-csp}

A finite-valued \emph{function} of arity $r=\ar{f}$ over a domain $D$ is a mapping $f : D^r \to \qplus$. A finite-valued \emph{constraint} over a set $X$ of variables is an expression of the form $f(\tuple{x})$, where $f$ is a finite-valued function and $\tuple{x} \in X^{\ar{f}}$. The set of variables appearing in $\tuple{x}$ is called the \emph{scope} of the constraint $f(\tuple{x})$. An instance $I$ of the \problem{Max-CSP} problem is a finite set $X = \{x_1,\ldots,x_n\}$ of variables, a finite domain $D$ of values, and an objective function of the form
\[
f_I(x_1,\ldots,x_n) = \sum_{i=1}^q f_i(\tuple{x_i})
\]
where each $f_i(\tuple{x_i})$, $1 \leq i \leq q$ is a finite-valued constraint. The goal is to compute the maximum value of $f_I$ over all possible assignments to
$X$, which we denote by $\opt{I}$. In this paper we assume that each function $f_i$, $1 \leq i \leq q$ is given in the input as the table of all pairs $(\tuple{d},f_i(\tuple{d}))$ where
$\tuple{d} \in D^{\ar{f_i}}$ and $f_i(\tuple{d}) > 0$ (the so-called \emph{positive representation}). 
It follows that the total size $\lVert I \rVert$ of a \problem{Max-CSP} instance $I$ is roughly
\[
    \sum_{i=1}^q  \bigg( \ar{f_i} \log(|X|) + \sum_{\substack{\tuple{d} \in D^{\ar{f_i}}\\ f_i(\tuple{d}) > 0}} ( \ar{f_i}\log(|D|) + |\text{enc}(f_i(\tuple{d}))| ) \bigg)
\]
where $\text{enc}(\cdot)$ is a reasonable encoding for rational numbers.

Actually, Max-CSPs are commonly defined with
only $\{0,1\}$-valued functions, or with $\{0,w\}$-valued functions, where $w$
could be different in different functions; the
latter are called weighted Max-CSPs. What we defined as Max-CSPs is a more
general framework, sometimes called \emph{finite-valued} CSPs~\cite{tz16:jacm}
or Max-CSPs with payoff functions~\cite{Raghavendra08:everycsp}.

The hypergraph of a \problem{Max-CSP} instance is the set of scopes of its constraints. 
Without loss of generality, we will always assume that no two constraints share
the same scope and for every constraint $f_i(\tuple{x_i})$, the entries of
$\tuple{x_i}$ are pairwise distinct. 
In particular, there is a bijection between the constraints of a \problem{Max-CSP} instance 
and the edges of its hypergraph.
Given a family $\mathcal{H}$ of hypergraphs, we denote by
\problem{Max-CSP($\mathcal{H},-$)} the restriction of \problem{Max-CSP} to the
instances whose hypergraph belongs to $\mathcal{H}$.

\section{Point decompositions and point-width}
\label{sec:pointdec}

\begin{figure*}[t]
\begin{center}
\begin{tikzpicture}
\usetikzlibrary{shapes}
\tikzstyle{tree edge}=[<-,very thick]
\tikzstyle{arc}=[->,black]
  \node (x0) at (-1,-1) {$x_0$};
  \node (x1) at (-1,0.9) {$x_1$};
  \node (x2) at (-2.414,-2.414) {$x_2$};
  \node (x3) at (0.414,-2.414) {$x_3$};

  \draw (x0) -- node [right] {$e_1$} (x1);
  \draw (x0) -- node [left,xshift=+0.6cm,yshift=0.0cm] {$e_2$} (x2);
  \draw (x0) -- node [right] {$e_3$} (x3);

  \draw [rounded corners=10mm] (-1,1.9) -- (-3.114,-3.114) -- (1.1141,-3.114) -- cycle;
  \node (e) at (-2.2,0) {$e$};

  \draw (4,0) ellipse (0.7cm and 1.5cm); 
  \node (t4) at (4,2.0) {$t_4$};
  \node (x400) at (4,0.9) {\footnotesize $\emptyset$};
  \node (x41) at (4,0.4) {\footnotesize $x_1$};
  \node (x42) at (4,0) {\footnotesize $x_2$};
  \node (x40) at (4,-0.5) {\footnotesize $x_0$};
  \node (x43) at (4,-1.0) {\footnotesize $x_3$};

  \node[shape=ellipse,draw,minimum size=19mm,minimum width=9mm] (sb41203) at (4,-0.25) {};

  \node at (4,-2) {\footnotesize $(x_1,e_1)$};
  \node at (4,-2.3) {\footnotesize $(x_0,e_1)$};
  \node at (4,-2.6) {\footnotesize $(x_1,e)$};
  \node at (4,-2.9) {\footnotesize $(x_2,e)$};
  \node at (4,-3.2) {\footnotesize $(x_0,e)$};
  \node at (4,-3.5) {\footnotesize $(x_3,e)$};

  \draw (6,0) ellipse (0.7cm and 1.5cm); 
  \node (t3) at (6,2.0) {$t_3$};
  \node (t300) at (6,0.8) {\footnotesize $\emptyset$};
  \node (t32) at (6,0.3) {\footnotesize $x_2$};
  \node (t30) at (6,-0.2) {\footnotesize $x_0$};
  \node (t3) at (6,-0.7) {\footnotesize $x_3$};

  \node[shape=ellipse,draw,minimum size=10mm,minimum width=7mm] (sb320) at (6,0.1) {};
  \node[shape=ellipse,draw,minimum size=17mm,minimum width=10mm] (sb3203) at (6,-0.2) {};

  \node at (6,-2) {\footnotesize $(x_2,e_2)$};
  \node at (6,-2.3) {\footnotesize $(x_0,e_2)$};
  \node at (6,-2.6) {\footnotesize $(x_2,e)$};
  \node at (6,-2.9) {\footnotesize $(x_0,e)$};
  \node at (6,-3.2) {\footnotesize $(x_3,e)$};

  \draw (8,0) ellipse (0.7cm and 1.5cm); 
  \node (t2) at (8,2.0) {$t_2$};
  \node[shape=circle,draw,inner sep=1pt] (sb20) at (8,0) {\footnotesize $x_0$};
  \node (t200) at (8,0.5) {\footnotesize $\emptyset$};
  \node (t23) at (8,-0.5) {\footnotesize $x_3$};

  \node[shape=ellipse,draw,minimum size=11mm,minimum width=7mm] (sb203) at (8,-0.25) {};

  \node at (8,-2) {\footnotesize $(x_0,e_1)$};
  \node at (8,-2.3) {\footnotesize $(x_0,e_2)$};
  \node at (8,-2.6) {\footnotesize $(x_0,e_3)$};
  \node at (8,-2.9) {\footnotesize $(x_3,e_3)$};
  \node at (8,-3.2) {\footnotesize $(x_0,e)$};
  \node at (8,-3.5) {\footnotesize $(x_3,e)$};

  \draw (10,0) ellipse (0.7cm and 1.5cm); 
  \node (t1) at (10,2.0) {$t_1$};
  \node (t100) at (10,0.3) {\footnotesize $\emptyset$};
  \node[shape=circle,draw,inner sep=1pt] (t13) at (10,-0.2) {\footnotesize $x_3$};

  \node at (10,-2) {\footnotesize $(x_3,e_3)$};
  \node at (10,-2.3) {\footnotesize $(x_3,e)$};

  \draw (12,0) ellipse (0.7cm and 1.5cm); 
  \node (t0) at (12,2.0) {$t_0$};
  \node (t000) at (12,0) {\footnotesize $\emptyset$};
  

  \draw[tree edge](4.7,0) -- (5.3,0);
  \draw[tree edge] (6.7,0) -- (7.3,0);
  \draw[tree edge] (8.7,0) -- (9.3,0);
  \draw[tree edge] (10.7,0) -- (11.3,0);

  \draw {(4,0.6) [rounded corners] -- (3.75,0.6) -- (3.75,-0.7) -- (4.2,-0.7) -- (4.2,-0.25) -- (3.8,-0.25) -- (3.8,0.2) -- (4.2,0.2) -- (4.2,0.6) -- (3.90,0.6)};


  \draw [arc] (sb41203) to [bend right=45] (sb3203);

  \draw[arc] (4.2,-0.6) to [bend right=45] (sb203);
  \draw[arc] (4.2,-0.5) to [bend right=0] (sb3203);
  \draw[arc] (4.2,+0.4) to [bend right=-45] (sb20);

  \draw[arc] (sb320) to [bend right=-35] (sb20);
  \draw[arc] (sb320) to [bend right=10] (sb203);
  \draw[arc] (sb3203) to [bend right=+30] (sb203);
  
  \draw[arc] (sb20) to [bend right=-35] (t13);
  \draw[arc] (sb20) to [bend right=-35] (t000);
  \draw[arc] (sb203) to [bend right=+35] (t13);

  \draw [arc] (t13) to [bend right=30] (t000);
\end{tikzpicture}
\end{center}
\caption{The hypergraph $H$ and its point decomposition from Examples~\ref{exa:first}--\ref{ex:beta}.}\label{fig:ex}
\end{figure*}
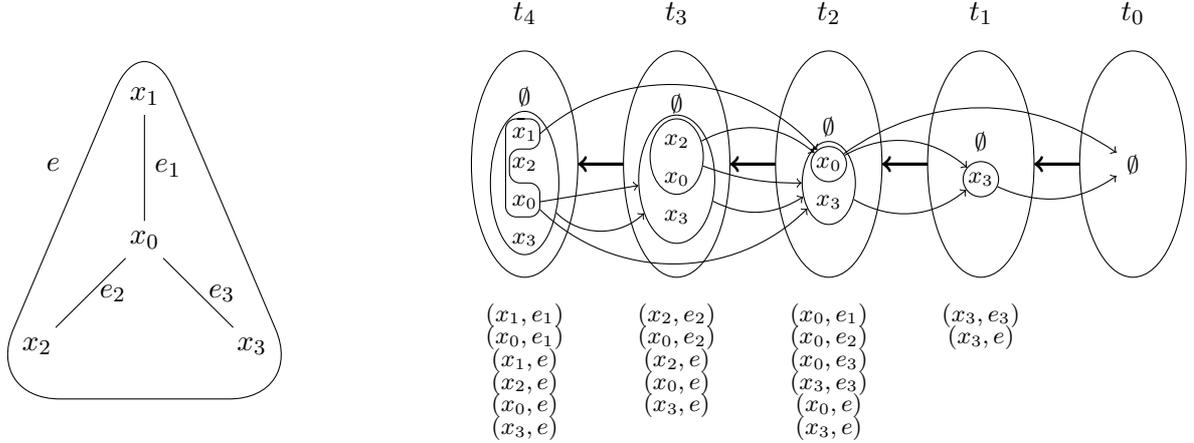

Let $H$ be a hypergraph. Let $\T=(T, (B_t)_{t\in V(T)})$ be a pair such that 
$T$ is a rooted tree and $B_t \subseteq \points{H}$ is a set of points, for every $t \in V(T)$. 
For $t\in V(T)$, we call the set $B_t$ the \emph{bag of $t$} and the pairs $(t, S)$ with $S \in
\lat{H|_{B_t}}$ the \emph{sub-bags of $t$}. 
We denote by $<_T$ the strict partial
order on $V(T)$ such that $t_1 <_T t_2$ if and only if $t_1$ is a
descendant of $t_2$ in $T$. 
A \emph{$\T$-structure} is a directed graph $A$ whose
vertex set is the set of all sub-bags of $V(T)$ and such that for every
arc $((t_1,S_1),(t_2,S_2))$ in $A$ we have $t_1 <_T t_2$. 

\begin{example}
\label{exa:first}
Consider the hypergraph $H=\{e, e_1, e_2, e_3\}$, where $e=\{x_0, x_1, x_2, x_3\}$ 
and $e_i = \{x_0, x_i\}$, for every $i\in\{1,2,3\}$; see Figure~\ref{fig:ex} on the left. 
In particular, $V(H)=\{x_0, x_1, x_2, x_3\}$. 
The right-hand side of Figure~\ref{fig:ex} depicts a pair $\T=(T,(B_t)_{t\in
V(T)})$, where $T$ is a path (depicted by bold arcs\footnote{We view the tree
$T$ as undirected although there is an implicit direction by the parent/child 
relationship. For clarity, in Figure~\ref{fig:ex}, we directed the
(bold) edges of the tree $T$ away from the root, which is $t_0$.}) rooted at $t_0$,
and the points in each bag $B_t$ are listed below each node.  
The sub-bags of each node of $T$ are depicted within the node.  
For instance, for the node $t_4$ we have $H|_{B_{t_4}}=\{\{x_1,x_0\},\{x_1,x_2,x_0,x_3\}\}$. 
Hence the sub-bags of $t_4$ are $(t_4,\emptyset)$, $(t_4,\{x_1,x_0\})$ and $(t_4,\{x_1,x_2,x_0,x_3\})$. 
The arcs between sub-bags represent a possible $\T$-structure $A$.  
\end{example}

\begin{definition}[Decomposability]
\label{def:decomp}
Let $A$ be a $\T$-structure for a pair $\T=(T, (B_t)_{t\in V(T)})$. We say that $A$ is
\emph{decomposable} if for any two arcs $(s_1,s)$, $(s_2,s)$ in $A$, if
\begin{enumerate}
\item[(i)] $s_1,s_2$ are sub-bags of different vertices of $V(T)$, and
\item[(ii)] there exist two sub-bags $s'_1, s'_2$ (not necessarily distinct) of the same
vertex $t \in V(T)$, and directed paths in $A$ from
$s'_1$ to $s_1$, and from $s'_2$ to $s_2$
\end{enumerate}
then either $(s_1,s_2) \in \edges{A}$ or $(s_2,s_1) \in \edges{A}$.
\end{definition}

Observe that if $A$ is \emph{not} decomposable due to arcs $(s_1,s)$, $(s_2,s)$, where 
$s_1,s_2$ are sub-bags of $t_1,t_2\in V(T)$, respectively, 
then either $t_1<_T t_2$ or $t_2<_T t_1$ must hold 
(otherwise, condition (ii) would fail). 
Let say that $t_1<_T t_2$. 
Note that it could be possible that $t=t_1$, 
in which case, the directed path from $s'_1$ to $s_1$ is simply the empty path, i.e., $s'_1=s_1$. 
If additionally, $s'_2=s_1$, we obtain the simplest case of non-decomposability, in which there is a directed path in $A$ from $s_1$ to $s_2$ (and $(s_1,s_2)\notin E(A)$).

\begin{example}
\label{exa:decomp}
The $\T$-structure $A$ from Example~\ref{exa:first} and Figure~\ref{fig:ex} is decomposable. Consider for instance the arcs $(s_1,s)$ and $(s_2,s)$ with $s=(t_2,\{x_0,x_3\})$, $s_1=(t_4,\{x_1,x_0\})$ and $s_2=(t_3,\{x_2,x_0,x_3\})$. We have that $s_1$ and $s_2$ are sub-bags of different vertices of $T$, and condition (ii) of decomposability holds if we take $s'_1 = s_1$ and $s'_2 = (t_4,\{x_0,x_1,x_2,x_3\})$. In this case decomposability requires that at least one of $(s_1,s_2)$ or $(s_2,s_1)$ is an arc of $A$, which is true for $(s_1,s_2)$.
\end{example}

The intuition behind decomposability is as follows. 
Suppose we have a sub-bag $s$ in the $\T$-structure and two incoming arcs $(s_1,s), (s_2,s)$ in $A$, where $s_1,s_2$ are sub-bags of distinct vertices $t_1,t_2\in V(T)$. 
Let $T_{s_1}$ be the set of nodes of $V(T)$ that can ``reach'' $s_1$, i.e., that contain a sub-bag $s'_1$ from which $s_1$ is reachable in $A$. 
Similarly, we define $T_{s_2}$. 
Then decomposability means that whenever $s_1$ and $s_2$ are ``incomparable'' with respect to $A$ (i.e., neither $(s_1,s_2)$ nor $(s_2,s_1)$ is an arc), 
then $T_{s_1}$ and $T_{s_2}$ must be disjoint. 

\begin{definition}[Realisations]
\label{def:real}
Let $A$ be a $\T$-structure for a pair $\T=(T, (B_t)_{t\in V(T)})$. 
 A \emph{realisation} of $A$ is a subgraph $A'$ of
$A$ induced by a subset $X \subseteq V(A)$ such that 
\begin{enumerate}
\item[(i)] $X$ contains at
most one sub-bag of each $t \in V(T)$, and 
\item[(ii)]  $A'$ has exactly one sink, which must
be a sub-bag of the root of $T$.
\end{enumerate}
\end{definition}

For any realisation $A'$ of a $\T$-structure $A$, we define
$T_{A'}$ as the rooted tree whose vertex set is 
\[V(T_{A'})=\{t\in V(T): \text{$\exists$ a sub-bag $(t,S)\in V(A')$}\},\]
and whose edges are defined as follows. 
Suppose $t_1,t_2\in V(T_{A'})$ due to sub-bags $(t_1,S_1), (t_2,S_2)\in V(A')$, respectively. 
Then $t_2$ is the parent of $t_1$, i.e., $(t_1,t_2)\in E(T_{A'})$, if $t_2$ is the least vertex with respect to $<_T$ of the set 
\[\{t\in V(T): \text{$\exists (t,S)\in V(A')$ and $((t_1,S_1),(t,S))\in E(A')$}\}.\]

\begin{example}
\label{exa:real}
For the $\T$-structure $A$ in Figure~\ref{fig:ex}, consider the subgraph $A_1$ of $A$ induced by the sub-bags 
$(t_4, \{x_1,x_0\})$, $(t_3,\{x_2,x_0,x_3\})$, $(t_2,\{x_0,x_3\})$, $(t_1,\{x_3\})$ and $(t_0,\emptyset)$. 
We have that $A_1$ is a realisation as the only sink is $(t_0,\emptyset)$. 
Note that if we remove from $A_1$ the sub-bag $(t_1,\{x_3\})$ then we obtain a subgraph that is not a realisation as now $(t_2,\{x_0,x_3\})$ becomes a sink. 
Observe also that $T_{A_1}$ is precisely $T$. Another possible realisation is the subgraph $A_2$ of $A$ induced by the sub-bags 
$(t_4, \{x_1,x_0\})$, $(t_3,\{x_2,x_0\})$, $(t_2,\{x_0\})$ and $(t_0,\emptyset)$. In this case, $T_{A_2}$ is the tree with vertices $\{t_0,t_2,t_3,t_4\}$ and edges $(t_2,t_0)$, $(t_3,t_2)$ and $(t_4,t_2)$.
\end{example}

For a $\T$-structure $A$ and a subhypergraph $H'$ of $H$, we denote by $A[H']$ 
the subgraph of $A$ induced by the set $\{(t,V(H'|_{B_t})): t\in V(T)\}$. 
We denote by $A[H']_\emptyset$ the directed graph obtained from $A[H']$ after removing every connected component $C$ in $A[H']$ that satisfies the following: 
for every sub-bag $(t,S)\in C$, we have that $t$ is not the root of $T$ and $S=\emptyset$. 
In other words, $A[H']_\emptyset$ contains precisely the connected components of $A[H']$ that contain a sub-bag of the root of $T$ or a sub-bag $(t,S)$ with $S\neq \emptyset$. 

\begin{example}
\label{exa:empty}
The subgraph $A_2$ of $A$ from Example~\ref{exa:real} is precisely $A[H']_\emptyset$, where $H'=\{e_1,e_2\}$. 
Note that $(t_1,\emptyset)$ needs to be removed from $A[H']$ in order to obtain $A[H']_\emptyset$. While $A[H']_\emptyset$ is a realisation, 
$A[H']$ is not, as $(t_1,\emptyset)$ is a sink. 
\end{example}

\begin{definition}[Point decomposition]
\label{def:point}
A \emph{point decomposition} of a hypergraph $H$ is a triple $(T, (B_t)_{t \in V(T)}, A)$ where $T$ is a rooted tree, each set $B_t\subseteq \points{H}$ is a set of points of $H$, $A$ is a decomposable $\T$-structure, 
where $\T=(T, (B_t)_{t \in V(T)})$, and
\begin{enumerate}
\item[(i)] For every edge $e \in H$, there exists $t \in V(T)$ such that $\points{\{e\}}=\{(v,e):v\in e\} \subseteq B_t$.
\item[(ii)] For every subhypergraph $H'$ of $H$, the subgraph $A[H']_\emptyset$ of $A$ is a realisation. 
\item[(iii)] For every realisation $A'$ of $A$ and $v \in \cup_{(t,S) \in V(A')} S$, the set 
\[
\{ t \in V(T_{A'}) : \text{$\exists (t,S)\in V(A')$ and $v\in S$}\}
\]
induces a connected subtree of $T_{A'}$.
\end{enumerate}
\end{definition}

A point decomposition is \emph{flat} if every arc in $A$ is between sub-bags of nodes adjacent in $T$. The \emph{width} of a point decomposition $(T, (B_t)_{t \in V(T)}, A)$ of a hypergraph $H$ is given by $\max_{t \in V(T)}\bcovernum{H|_{B_t}}$, the \emph{point-width} of $H$, denoted by $\pw{H}$, is the minimum width over all its point decompositions, and the \emph{flat point-width} of $H$, denoted by $\fpw{H}$, is the minimum width over all its flat point decompositions.

Throughout the paper we assume a straightforward encoding for point decompositions, where each bag is given as a list of points, the tree $T$ is given as a rooted graph whose vertex set is the set of all bags, and the $\T$-structure $A$ is given as a directed graph whose vertex set is the set of all sub-bags. We denote by $\lVert P \rVert$ the encoding size of a point decomposition $P$. We remark that checking whether a triple $(T, (B_t)_{t \in V(T)}, A)$ is a point decomposition may be a difficult task due to conditions (ii) and (iii). Whether it can be done in polynomial time is an interesting question, which we leave for future work.

\begin{example}
\label{exa:point}
Figure~\ref{fig:ex} shows a point decomposition of the hypergraph $H$ to the left. Note that $\bcovernum{H|_{B_{t_i}}}=1$, for $1\leq i\leq 4$, and then
the width of the decomposition is $1$. Hence $\pw{H}=1$. Note that the decomposition is not flat.  
\end{example}

As mentioned in the introduction, the intuition is that a $\T$-structure $A$ in a point decomposition of width $k$  
encodes various tree decompositions of hypertreewidth at most $k$ (cf.
Appendix~\ref{sec:widths} for a precise definition of tree decomposition and hypertreewidth), 
and in particular, one for each subhypergraph $H'$ of $H$. 
Such a tree decomposition for $H'$ is given by the tree $T_{A[H']_\emptyset}$ and the bags correspond to the sub-bags in $A[H']_\emptyset$. 

Finally, let us remark that once we know the $\T$-structure of a point decomposition, the particular form of the tree $T$ is irrelevant.  
Indeed, we can always assume that $T$ is a path: if it is not the case, we can extend $<_T$ to a total order $<_{\text{tot}}$ on $\vertices{T}$, which is precisely $<_{T'}$ for a certain path $T'$, and then replace $T$ by $T'$ in the point decomposition.
However, in the case of \emph{flat} point decompositions this is not true.
Hence, in general, we shall not impose any assumption on the tree $T$. 

\section{The algorithm}
\label{sec:alg}

In this section we describe a polynomial-time algorithm for solving \problem{Max-CSPs} when the input instance is paired with a point decomposition of bounded width of its hypergraph. We start with a number of simple definitions and observations before proving the main result in Theorem~\ref{thm:alg}.

\begin{definition}[Partial realisations]
Let $H$ be a hypergraph and $(T, (B_t)_{t \in \vertices{T}}, A)$ be a point decomposition of $H$. A \emph{partial realisation} of $A$ is a subgraph $A'$ of $A$ induced by a subset $X \subseteq \vertices{A}$ such that (i) $X$ contains at most one sub-bag of each $t \in \vertices{T}$, (ii) $A'$ has exactly one sink $s$ and (iii) there is a (possibly empty) directed path in $A$ from $s$ to a sub-bag of the root of $T$.
\end{definition}

The rooted tree $T_{A'}$ of a partial realisation $A'$ is defined the same way as for realisations: its vertex set is the set of all $t \in \vertices{T}$ with at least one sub-bag in $\vertices{A'}$, and the parent of $t_1 \in \vertices{T_{A'}}$ with $(t_1,S_1) \in \vertices{A'}$ is the least vertex with respect to $<_T$ in the set $\{t\in V(T): \text{$\exists (t,S)\in V(A')$ and $((t_1,S_1),(t,S))\in E(A')$}\}$. The next observation is a minor extension of condition (iii) of point decompositions to partial realisations.

\begin{observation}
\label{obs:connpartial}
Let $H$ be a hypergraph, $(T, (B_t)_{t \in \vertices{T}}, A)$ be a point decomposition of $H$, $A'$ be a partial realisation of $A$ and $v \in \cup_{(t,S) \in \vertices{A'}}S$. Then, the set \[
\{ t \in V(T_{A'}) : \text{$\exists (t,S)\in V(A')$ and $v\in S$}\}
\]
induces a connected subtree of $T_{A'}$.
\end{observation}

\begin{proof}
Let $s$ be the unique sink of $A'$. If $s$ is a sub-bag of the root of $T$ then $A'$ is a realisation and the claim follows from condition (iii) of point decompositions. Otherwise, let $(s,s_1,\ldots,s_n)$ be a directed path in $A$ from $s$ to a sub-bag $s_n$ of the root of $T$. The subgraph $A^*$ of $A$ induced by $\vertices{A'} \cup \{s_1,\ldots,s_n\}$ is a realisation and $T_{A'}$ is precisely the subtree of $T_{A^*}$ rooted at $s$, so the observation follows.
\end{proof}

\begin{definition}[Guards]
\label{def:guard}
Let $H$ be a hypergraph, $(T, (B_t)_{t \in \vertices{T}}, A)$ be a point decomposition of $H$ and $(t,S)$ be a sub-bag of $t \in \vertices{T}$. A \emph{guard} of $(t,S)$ is an inclusion-minimal subhypergraph $H'$ of $H$ such that $\vertices{H'|_{B_t}} = S$.
\end{definition}

Given a \problem{Max-CSP} instance $I$ with
hypergraph $H$ and $e \in H$, we will use $f_e(\tuple{x_e})$ to denote the
unique constraint with scope $e$. (As usual $X$ and $D$ denote the variables and the domain of $I$, respectively.) 
Given a constraint $f_e(\tuple{x_e})$ with $e
\in H$, its \emph{support} is the relation $R_e \defeq \{  \tuple{d} \in D^{|e|}
: f_e(\tuple{d}) > 0\}$. Without ambiguity we will sometimes treat $R_e$ as a set
of assignments to $e$. As usual, for an assignment $\psi$ with domain $Y$ and a subset $Y'\subseteq Y$, 
we denote by $\psi|_{Y'}$ the restriction of $\psi$ to $Y'$. Similarly, for a set $R$ of assignments over $Y$, 
we denote by $R|_{Y'}$ the set $\{\psi|_{Y'}: \psi\in R\}$.  
If $\psi:X' \to D$ is an assignment to $X' \subseteq
X$, we define $\tvalue{\psi} = \sum_{e \in H : e \subseteq X'}
f_e(\psi(\tuple{x_e}))$ and call $\psi$ a \emph{partial assignment to $X$}. In
particular, for any partial assignment $\psi$ to $X$, we have that $\tvalue{\psi} \leq
\opt{I}$. 

Given a partial assignment $\psi:X' \to D$, we say that
$\psi$ \emph{satisfies} an edge $e \in H$ if $\psi|_{X' \cap e} \in R_e|_{X' \cap e}$, and satisfies a subhypergraph if it satisfies all of its edges. Note that $\psi$ can satisfy edges that are not completely contained in $X'$. 
For $1\leq i\leq n$, with $n\geq 2$, let $R_i$ be a set of partial assignments from $X_i\subseteq X$ to $D$. 
The \emph{join} of $R_1,\dots,R_n$ is the set of all partial assignments $\psi:\bigcup_{i=1}^n X_i\to D$ 
such that $\psi|_{X_i}\in R_i$, for every $1\leq i\leq n$. 
Observe that a partial assignment $\psi:X' \to D$ satisfies a subhypergraph $H'\subseteq H$ if and only if 
$\psi$ restricted to $\bigcup_{e\in H'}(X'\cap e)$ belongs to the join of $\{R_e|_{X'\cap e}\}_{e\in H'}$.

\begin{definition}[Consistent assignments]
\label{def:sbcons}
Let $H$ be the hypergraph of a \problem{Max-CSP} instance and $(T, (B_t)_{t \in \vertices{T}}, A)$ be a point decomposition of $H$. If $s=(t,S)$ is a sub-bag of $t \in \vertices{T}$, an \emph{$s$-valid assignment} is an assignment $\psi : S \to D$ such that $\psi$ satisfies some guard $C$ of $s$. A \emph{consistent assignment to a partial realisation $A'$} of $A$ is a function $\phi$ that maps every sub-bag $s=(t,S) \in \vertices{A'}$ to an $s$-valid assignment such that for any two sub-bags $(t_1,S_1)$, $(t_2,S_2)$ with $t_1,t_2$ adjacent in $T_{A'}$, $\phi((t_1,S_1))|_{S_1 \cap S_2} = \phi((t_2,S_2))|_{S_1 \cap S_2}$.
\end{definition}

The following is a direct consequence from Observation~\ref{obs:connpartial}.

\begin{observation}
\label{obs:localtoglobal}
Let $H$ be the hypergraph of a \problem{Max-CSP} instance, $(T, (B_t)_{t \in \vertices{T}}, A)$ be a point decomposition of $H$, $\phi$ be a consistent assignment to some partial realisation $A'$ of $A$ and $X' \defeq \cup_{(t,S) \in \vertices{A'}} S$. Then, there exists an assignment $\psi : X' \to D$ such that for every $s=(t,S) \in \vertices{A'}$, $\phi(s) = \psi|_S$.
\end{observation}

\begin{definition}
Let $H$ be the hypergraph of a \problem{Max-CSP} instance, $(T, (B_t)_{t \in \vertices{T}}, A)$ be a point decomposition of $H$, $\phi$ be a consistent assignment to a partial realisation $A'$ of $A$ and $\psi$ be as in Observation~\ref{obs:localtoglobal}. The \emph{value} of $(\phi,A')$ is the quantity
\[
\tvalue{\phi,A'} \defeq \sum_{e \in H : \exists (t,S) \in \vertices{A'}, \, e \subseteq S} f_e(\psi(\tuple{x_e})).
\]
\end{definition}
The general idea behind the algorithm is to traverse the tree $T$ of the point decomposition bottom-up, keeping track for each sub-bag $s$ and $s$-valid assignment $\psi$ of the best value achievable by a partial realisation $A'$ with sink $s$ and consistent assignment to $A'$ that agrees with $\psi$ on $s$. The fact that $A$ is decomposable ensures that joining multiple partial realisations to a common sink always produces a partial realisation, as long as their initial sinks form an independent set in a certain (easily computable) chordal graph. This property enables a dynamic programming approach. It will follow from conditions (i), (ii) and (iii) in the definition of point decompositions that the maximum of the values computed by this algorithm at the root of $T$ is, in fact, the optimum of the \problem{Max-CSP} instance.

\begin{proposition}
\label{prop:realopt}
Let $I$ be a \problem{Max-CSP} instance with hypergraph $H$ and $(T, (B_t)_{t \in \vertices{T}}, A)$ be a point decomposition of $H$. The maximum of $\tvalue{\phi,A'}$ over all realisations $A'$ of $A$ and consistent assignments $\phi$ to $A'$ is exactly $\opt{I}$.
\end{proposition}

\begin{proof}
Let $M$ be the maximum of $\tvalue{\phi,A'}$ over all realisations $A'$ of $A$ and consistent assignments $\phi$ to $A'$. 

We first prove $M \geq \opt{I}$. Let $\psiopt$ be an assignment to the variables of $I$ such that $\tvalue{\psiopt} = \opt{I}$, and let $H' \subseteq H$ be the set of edges satisfied by $\psiopt$. Consider the subgraph $A[H']_{\emptyset}$ of $A$, which by condition (ii) of point decompositions is a realisation. We define $\phi^*$ as the function that maps each $(t,S) \in \vertices{A[H']_{\emptyset}}$ to $\psiopt|_S$. Since $\psiopt$ satisfies $H'$, it satisfies at least one guard for each sub-bag $(t,S) \in \vertices{A[H']_{\emptyset}}$. Therefore, $\phi^*$ is a consistent assignment to $A[H']_{\emptyset}$. By condition (i) of point decompositions, for every edge $e \in H'$ there exists $(t,S) \in \vertices{A[H']_{\emptyset}}$ such that $e \subseteq S$, and hence $M \geq \tvalue{\phi^*,A[H']_{\emptyset}} = \opt{I}$.

We now prove $\opt{I} \geq M$. Let $A'$ be a realisation of $A$ and $\phi$ be a consistent assignment to $A'$ such that $\tvalue{\phi,A'} = M$. 
By Observation~\ref{obs:localtoglobal}, there exists an assignment $\psi$ to $X' \defeq \cup_{(t,S) \in \vertices{A'}} S$ such that 
\[
    \tvalue{\psi}  = \sum_{e \in H : e \subseteq X'} f_e(\psi(\tuple{x_e}))\ \geq 
    \sum_{e \in H : \exists (t,S) \in \vertices{A'}, \, e \subseteq S} f_e(\psi(\tuple{x_e})) = \tvalue{\phi,A'} = M
\]
and hence $\opt{I} \geq M$.
\end{proof}

If $A'$ is a partial realisation and $s \in \vertices{A'}$, we use $A'[s]$ to denote the partial realisation induced by the sub-bags $s'$ of $A'$ such that there is a (possibly empty) directed path in $A'$ from $s'$ to $s$.

\begin{observation}
\label{obs:branches}
Let $H$ be the hypergraph of a \problem{Max-CSP} instance, $(T, (B_t)_{t \in \vertices{T}}, A)$ be a point decomposition of $H$, $\phi$ be a consistent assignment to a partial realisation $A'$ of $A$ with sink $s=(t,S)$ and $\psi$ be as in Observation~\ref{obs:localtoglobal}. Let $W$ be the set of all sub-bags $s'=(t',S')$ in $\vertices{A'}$ such that $t'$ is a child of $t$ in $T_{A'}$. Then,
\[
        \tvalue{\phi,A'} = \sum_{e \in H: e \subseteq S}f_e(\psi(\tuple{x_e}))
      +
   \sum_{\substack{s' \in W\\ s'=(t',S')}} \left(\tvalue{\phi|_{\vertices{A'[s']}},A'[s']} - \sum_{e \in H: e \subseteq S \cap S'}f_e(\psi(\tuple{x_e})) \right).
\]
\end{observation}

\begin{proof}
By definition of $T_{A'}$ there is no arc $(s_1,s_2)$ in $A$ with $s_1,s_2 \in W$. Since $A$ is decomposable, it follows that the sets $\vertices{A'[s']}$, $s' \in W$, are pairwise disjoint. Furthermore, by Observation~\ref{obs:connpartial}, if there exist an edge $e \in H$ and two
  sub-bags $s_1,s_2 \in W$ with $e \subseteq \left ( \cup_{(t^*,S^*) \in
  \vertices{A'[s_1]}} S^* \right) \cap \left( \cup_{(t^*,S^*) \in
  \vertices{A'[s_2]}} S^* \right)$ then $e \subseteq S$. Similarly, if there exist $e \in H$ and $s_1=(t_1,S_1) \in W$ such that $e \subseteq \left( \cup_{(t^*,S^*) \in \vertices{A'[s_1]}} S^* \right) \cap S$, then $e \subseteq S_1$. Putting everything together we have

\begin{align*}
\tvalue{\phi,A'} &= \sum_{e \in H : \exists (t^*,S^*) \in \vertices{A'}, \, e \subseteq S^*} f_e(\psi(\tuple{x_e}))\\
&= \sum_{e \in H : e \subseteq S} f_e(\psi(\tuple{x_e})) + 
\sum_{s' \in W} \left( \sum_{e \in H, e \not\subseteq S : \exists (t^*,S^*) \in \vertices{A'[s']}, \, e \subseteq S^*}  f_e(\psi(\tuple{x_e})) \right)\\
&= \sum_{e \in H: e \subseteq S}f_e(\psi(\tuple{x_e})) + 
\sum_{\substack{s' \in W\\ s'=(t',S')}} \left(\tvalue{\phi|_{\vertices{A'[s']}},A'[s']} - \sum_{e \in H: e \subseteq S \cap S'}f_e(\psi(\tuple{x_e})) \right)
\end{align*}
as claimed.
\end{proof}

Recall that an independent set in a graph is a subset of vertices that induces a subgraph with no edges. We will denote by $\is{G}$ the set of all independent sets in a graph $G$. 

\begin{theorem}
\label{thm:alg}
Let $k$ be a fixed positive integer. There exists an algorithm which, given as input a \problem{Max-CSP} instance $I$ with hypergraph $H$ and a point decomposition $P=(T, (B_t)_{t \in \vertices{T}}, A)$ of $H$ of width at most $k$, computes $\opt{I}$ in time polynomial in $\Vert P \Vert$ and $\Vert I \Vert$.
\end{theorem}

\begin{proof}
We first describe the algorithm. To each bag $t \in \vertices{T}$, sub-bag $s=(t,S)$ and $s$-valid assignment $\psi$ we will associate a nonnegative rational value $\tvaluealg{s,\psi}$. We will compute these values bottom-up, starting from the leaves of $T$.

Let $t$ be a vertex of $T$, $s=(t,S)$ be a sub-bag of $t$ and $\psi$ be an $s$-valid assignment. Suppose that the values $\tvaluealg{s',\psi'}$ have already been computed for all pairs $(s'=(t',S'),\psi')$ with $t' <_T t$. If $t$ is a leaf then we set $\tvaluealg{s,\psi} \defeq \sum_{e \in H : e \subseteq S} f_e(\psi(\tuple{x_e}))$. If $t$ is not a leaf then we define a vertex-weighted graph $G$ where
\begin{itemize}
\item $\vertices{G}$ is the set of all sub-bags $s' = (t',S')$ with $t' <_T t$ such that (i) there exists at least one $s'$-valid assignment $\psi'$ such that $\psi'|_{S \cap S'} = \psi|_{S \cap S'}$ and (ii) $(s',s)$ is an arc in $A$;
\item $\edges{G}$ is the set of all pairs $\{(t_1,S_1),(t_2,S_2)\} \in \vertices{G}^2$ such that either $t_1 = t_2$ or $((t_1,S_1),(t_2,S_2))$ is an arc in $A$;
\item For every $s' = (t',S') \in \vertices{G}$, the weight $w(s')$ of $s'$ is the maximum of 
\[
  \tvaluealg{s',\psi'} - \sum_{e \in H: e \subseteq S \cap S'}f_e(\psi(\tuple{x_e}))
\]
over all $s'$-valid assignments $\psi'$ such that $\psi'|_{S \cap S'} = \psi|_{S \cap S'}$. Note that this quantity is well-defined because at least one suitable assignment $\psi'$ exists, by definition of $V(G)$.
\end{itemize}
We then set $\tvaluealg{s,\psi} \defeq \sum_{e \in H: e \subseteq S}f_e(\psi(\tuple{x_e})) + \max_{U \in \is{G}} \left( \sum_{s' \in U}w(s') \right)$. Once $\tvaluealg{s,\psi}$ is computed for all pairs $(s,\psi)$ where $s$ is a sub-bag of the root of $T$, the algorithm outputs the maximum of $\tvaluealg{s,\psi}$ over all such pairs.

\begin{claim}
\label{claim:algo1}
For every $t \in \vertices{T}$, sub-bag $s=(t,S)$ with a (possibly empty) directed path in $A$ from $s$ to a sub-bag of the root of $T$ and $s$-valid assignment $\psi$, $\tvaluealg{s,\psi}$ is the maximum of $\tvalue{\phi,A'}$ over all partial realisations $A'$ of $A$ whose sink is $s$ and consistent assignments $\phi$ to $A'$ such that $\phi(s) = \psi$.
\end{claim}

\begin{proof}
We proceed by induction, proving the claim for all pairs $(s,\psi)$ in the same order the algorithm computes $\tvaluealg{s,\psi}$. Let $s=(t,S)$ be a sub-bag with a directed path in $A$ to a sub-bag of the root of $T$ and $\psi$ be an $s$-valid assignment. Suppose that the claim holds for all pairs $(s',\psi')$ for which $\tvaluealg{s',\psi'}$ is computed by the algorithm before  $\tvaluealg{s,\psi}$ (and in particular for all pairs $(s',\psi')$ where $s'$ is a sub-bag of $t'$ with $t' <_T t$). If $t$ is a leaf then the claim trivially holds, so suppose that $t$ is not a leaf. 
We start by showing that $\tvaluealg{s,\psi}$ is at least the maximum over all $\tvalue{\phi,A'}$. 
Let $A'$ be any partial realisation of $A$ with sink $s$ and $\phi$ be a consistent assignment to $A'$ with $\phi(s) = \psi$. Let $W$ be the set of all sub-bags $s'=(t',S')$ in $\vertices{A'}$ such that $t'$ is a child of $t$ in $T_{A'}$. Note that we have $t' <_T t$ for all $(t',S')$ in $W$; it follows from the definition of the tree $T_{A'}$ that there does not exist an arc $((t',S'),(t'',S''))$ in $A$ with $(t',S'),(t'',S'') \in W$ (as otherwise one of $t',t''$ would not have $t$ as parent in $T_{A'}$). Therefore, $W$ is a subset of $\vertices{G}$ and forms an independent set. By Observation~\ref{obs:branches} and the induction hypothesis we have

\begin{align*}
\tvalue{\phi,A'} &= \sum_{e \in H: e \subseteq S}f_e(\psi(\tuple{x_e})) + 
\sum_{\substack{s' \in W\\s'=(t',S')}} \left(\tvalue{\phi|_{\vertices{A'[s']}},A'[s']} - \sum_{\substack{e \in H:\\ e \subseteq S \cap S'}}f_e(\psi(\tuple{x_e})) \right)\\
&\leq \sum_{e \in H: e \subseteq S}f_e(\psi(\tuple{x_e})) + 
\sum_{\substack{s' \in W\\s'=(t',S')}} \left(\tvaluealg{s',\phi(s')} - \sum_{\substack{e \in H:\\ e \subseteq S \cap S'}}f_e(\psi(\tuple{x_e})) \right).
\end{align*}
Then, from the definition of the vertex weights in $G$ we deduce
\begin{align*}
\tvalue{\phi,A'} &\leq \sum_{e \in H: e \subseteq S}f_e(\psi(\tuple{x_e})) + \sum_{s'=(t',S') \in W} w(s')
\end{align*}
and since $\tvaluealg{s,\psi}$ is the maximum of the right-hand side expression taken over all independent sets $W'$ of $G$, we finally obtain that $\tvalue{\phi,A'} \leq \tvaluealg{s,\psi}$, as required.

For the other direction, we need only prove that there exist a partial realisation $A'$ with sink $s$ and a consistent assignment $\phi$ to $A'$ such that $\phi(s) = \psi$ and $\tvalue{\phi,A'}$ is exactly $\tvaluealg{s,\psi}$. Let $W$ be the independent set of $G$ chosen by the algorithm to compute $\tvaluealg{s,\psi}$. For each sub-bag $s'=(t',S') \in W$, let $\psi_{s'}$ be an $s'$-valid assignment such that $\tvaluealg{s',\psi_{s'}} - \sum_{e \in H: e \subseteq S \cap S'}f_e(\psi(\tuple{x_e})) = w(s')$ and $\psi_{s'}|_{S \cap S'} = \psi|_{S \cap S'}$. Note that every sub-bag in $W$ can reach a sub-bag of the root of $T$ via a directed path in $A$ by going through $s$. Then, by induction for each $s' \in W$ there exist a partial realisation $A'_{s'}$ with sink $s'$ and a consistent assignment $\phi_{s'}$ to $A'_{s'}$ such that $\phi_{s'}(s') = \psi_{s'}$ and $\tvalue{\phi_{s'},A'_{s'}} = \tvaluealg{s',\psi_{s'}} = w(s') + \sum_{e \in H: e \subseteq S \cap S'}f_e(\psi(\tuple{x_e}))$. Now, if we define $A'$ as the subgraph of $A$ induced by $\{s\} \cup \left( \cup_{s' \in W} \vertices{A'_{s'}} \right)$, then (i) $A'$ has a single sink $s$, since the sinks of each $A'_{s'}$ have an outgoing arc to $s$, and (ii) $A'$ contains at most one sub-bag for each $t \in \vertices{T}$ because $A$ is decomposable and $W$ is an independent set in $G$. It follows that $A'$ is a partial realisation of $A$. 

The mapping $\phi$ defined on $\vertices{A'}$ such that $\phi(s^*) \defeq \psi$ if $s^* = s$ and $\phi(s^*) \defeq \phi_{s'}(s^*)$ otherwise, where $s'$ is the only sub-bag in $W$ such that $s^* \in \vertices{A'_{s'}}$, is a consistent assignment to $A'$. Finally, by Observation~\ref{obs:branches} and the induction hypothesis we obtain
\begin{align*}
\tvalue{\phi,A'} &= \sum_{e \in H: e \subseteq S}f_e(\psi(\tuple{x_e})) +
   \sum_{\substack{s' \in W\\ s'=(t',S')}} \left(\tvalue{\phi_{s'},A'_{s'}} - \sum_{e \in H: e \subseteq S \cap S'}f_e(\psi(\tuple{x_e})) \right)\\
   &= \sum_{e \in H: e \subseteq S}f_e(\psi(\tuple{x_e})) + \sum_{s' \in W} w(s')
\end{align*}
which is exactly $\tvaluealg{s,\psi}$.
\renewcommand{\qedsymbol}{$\blacksquare$}\end{proof}

\begin{corollary}
\label{cor:optalg}
The output of the algorithm is the maximum of $\tvalue{\phi,A'}$ over all realisations $A'$ of $A$ and consistent assignments $\phi$ to $A'$.
\end{corollary}

We deduce from Corollary~\ref{cor:optalg} and Proposition~\ref{prop:realopt} that the algorithm correctly outputs $\opt{I}$. We now turn to the problem of estimating the time complexity of the algorithm. To this end, we will need to bound the time necessary to compute the maximum-weight independent sets. This will be achieved with the help of the next claim.

A graph is \emph{chordal} if every cycle $C$ with at least four vertices has a \emph{chord}, that is, an edge connecting two vertices that are not consecutive in $C$.

\begin{claim}
\label{claim:chordal}
For any given pair $(s,\psi)$, the associated graph $G$ is chordal.
\end{claim}

\begin{proof}
By way of contradiction let us assume that there exists a pair $(s,\psi)$ for which $G$ has a chordless cycle $C$. 
Let $s_1=(t_1,S_1)$ be a sub-bag in $C$ such that $t_1$ is minimal with respect to $<_T$. Since $C$ is chordless, at least one of the two sub-bags that are adjacent to $s_1$ in $C$ is not a sub-bag of $t_1$. Let $s_2$ be that sub-bag, and $s_3$ be the other one. Note that $s_2$ and $s_3$ are not adjacent in $G$, which means that they are not sub-bags of the same vertex of $T$ and none of $(s_2,s_3),(s_3,s_2)$ is an arc in $A$. Furthermore, since $t_1$ is minimal with respect to $<_T$ in the cycle, there is a directed path (of length $1$) in $A$ from $s_1$ to $s_2$. Likewise, there is always a directed path in $A$ from some sub-bag of $t_1$ to $s_3$: if $s_3$ is a sub-bag of $t_1$ then this path is empty, and otherwise we have the path $(s_1,s_3)$ in $A$ by minimality of $t_1$. Finally, by construction we have the arcs $(s_2,s)$ and $(s_3,s)$ in $A$, so the triple $(s,s_2,s_3)$ contradicts the decomposability of $A$. Thus the chordless cycle $C$ does not exist, which establishes the claim.
\renewcommand{\qedsymbol}{$\blacksquare$}\end{proof}

\begin{claim}
\label{claim:runtime}
The runtime of the algorithm is polynomial in $\Vert I \Vert$ and $\Vert P \Vert$.
\end{claim}

\begin{proof}
By definition of the width of a point decomposition, for each bag $B_t$, $t \in \vertices{T}$ we have $\bcovernum{H|_{B_t}} \leq k$. Hence, for each subhypergraph $H' \subseteq H$ there exists a subhypergraph $H^* \subseteq H'$, $|H^*| \leq k$, such that $\vertices{H^*|_{B_t}} = \vertices{H'|_{B_t}}$; in particular, every guard of a sub-bag contains at most $k$ edges. Therefore, given a sub-bag $s$, any $s$-valid assignment is in the join of restrictions of the support of at most $k$ constraints; it follows that there are at most $|H|^k q^k$ distinct $s$-valid assignments, where $q \defeq \max_{e \in H}|R_e|$, and the algorithm computes $\tvaluealg{s,\psi}$ for $O(\Vert P \Vert |H|^k q^k)$ pairs $(s,\psi)$.

The computation of $\tvaluealg{s,\psi}$ for a given pair $(s,\psi)$ reduces to computing a maximum weighted independent set in the graph $G$, which can be achieved in time linear in $\Vert G \Vert = O(\Vert P \Vert)$ since $G$ is chordal~\cite{frank1975some,Tarjan85:dm} by Claim~\ref{claim:chordal}. Constructing the graph $G$ takes time polynomial in $\Vert P \Vert$ and $|H|^k q^k$, which concludes the proof of Claim~\ref{claim:runtime}.
\renewcommand{\qedsymbol}{$\blacksquare$}\end{proof}

Theorem~\ref{thm:alg} now follows from Corollary~\ref{cor:optalg}, Proposition~\ref{prop:realopt} and Claim~\ref{claim:runtime}.
\end{proof}

\section{Relationship with $\beta$-acyclicity}
\label{sec:beta}

A hypergraph $H$ is \emph{$\alpha$-acyclic}~\cite{Beeri83:jacm} if it has a \emph{join tree}. 
A join tree is a pair $(T,\lambda)$ where $T$ is a tree and $\lambda$ is a bijection from $V(T)$ to (the edges of) $H$, 
such that for every $v\in V(H)$ the set $\{t\in V(T): v\in \lambda(t)\}$ induces a connected subtree of $T$. 
A hypergraph $H$ is \emph{$\beta$-acyclic}~\cite{Fagin83:jacm-degrees} if every subhypergraph of $H$ is $\alpha$-acyclic. 
It is known that $\beta$-acyclic hypergraphs are tractable for \problem{Max-CSPs}:

\begin{theorem}[\cite{brault-baron15:stacs}]
\label{theo:beta-BCM}
\problem{Max-CSP($\mathcal{H},-$)} can be solved in polynomial time if $\mathcal{H}$ is
a family of $\beta$-acyclic hypergraphs.
\end{theorem}

The algorithm of Brault-Baron, Capelli, and Mengel~\cite{brault-baron15:stacs} works
by variable elimination, making use of a well-known alternative characterisation
of $\beta$-acyclic hypergraphs in terms of the so-called \emph{$\beta$-elimination orders}~\cite{Beeri83:jacm}.
In this section we show that such hypergraphs are covered by our framework as they always have a point decomposition of polynomial size and width $1$, which can be computed in polynomial time. 
Hence, together with Theorem \ref{thm:alg}, we can obtain Theorem \ref{theo:beta-BCM}. 

An ordering $(x_1,\ldots,x_n)$ of the vertices of a hypergraph $H$ is a \emph{$\beta$-elimination order} if for any $x_i \in \vertices{H}$ and $e,e' \in H$ such that $x_i \in e \cap e'$, either $e \cap \{x_j : j \geq i\} \subseteq e'$ or $e' \cap \{x_j : j \geq i\} \subseteq e$. 
A hypergraph is $\beta$-acyclic if and only if it has a $\beta$-elimination order~\cite{Beeri83:jacm}.

Our construction of point decompositions for $\beta$-acyclic hypergraphs is
inspired by recent work of Capelli~\cite{capelli17:lics}, from whom we borrow
some notation and lemmas. Let $H$ be a $\beta$-acyclic hypergraph and
$<_{\beta}$ be a $\beta$-elimination order of $H$. Given a vertex $x \in
\vertices{H}$, let $\vertices{H}_{\leq x}:=\{v\in V(H):v\leq_{\beta} x\}$ 
and $\vertices{H}_{\geq x}:=\{v\in V(H):v\geq_{\beta} x\}$. 
Let $<_H$ be the total order on the edges of $H$ such that $e_1 <_H e_2$ if and only
if $\max_{<_{\beta}} (e_1 \Delta e_2) \in e_2$, where $\Delta$ denotes the
symmetric difference. A \emph{walk} from $e \in H$ to $f \in H$ is a sequence $(e_1,x_1,e_2,x_2,\ldots,x_{n-1},e_n)$, with $n\geq 1$, 
where each $e_i$ is an edge of $H$, $e_1 = e$, $e_n = f$, and each $x_i$ is a vertex of $H$ such that $x_i \in e_i \cap e_{i+1}$. 
Given $x \in \vertices{H}$ and $e \in H$, let $H^x_e$ denote the set of edges of $H$ reachable from $e$ through a walk that contains only vertices $\leq_\beta x$ and edges $\leq_H e$.

\begin{example}
\label{ex:beta}
Consider the hypergraph $H$ from Figure~\ref{fig:ex} defined as $H=\{e, e_1, e_2, e_3\}$, where $e=\{x_0, x_1, x_2, x_3\}$ and
  $e_i = \{x_0, x_i\}$, for $i\in\{1,2,3\}$. We have that $H$ is $\beta$-acyclic. A possible $\beta$-elimination order is $x_1 <_{\beta} x_2 <_{\beta} x_0 <_{\beta} x_3$. 
  The induced order $<_H$ is $e_1 <_H e_2 <_H e_3 <_H e$. 
  For instance, note that $e_1\not\in H^{x_2}_{e_3}$ as the only possible walk would be $(e_3,x_0,e_1)$ but $x_0 >_\beta x_2$. 
  We have $H^{x_2}_{e_3}=\{e_3\}$ and $H^{x_0}_{e_3}=\{e_3, e_1,e_2\}$. Note that $e\not\in H^{x_0}_{e_3}$ as $e>_H e_3$. 
\end{example}

\begin{lemma}[\protect{\cite[Lemma~2]{capelli17:lics}}]
\label{lem:capelli2}
Let $x,y \in \vertices{H}$ such that $x \leq_{\beta} y$ and $e,f \in H$ such that $e \leq_H f$ and $\vertices{H^x_e} \cap \vertices{H^y_f} \cap \vertices{H}_{\leq x} \neq \emptyset$. Then, $H^x_e \subseteq H^y_f$.
\end{lemma}

\begin{theorem}[\protect{\cite[Theorem~3]{capelli17:lics}}]
\label{thm:capelli3}
For every $x \in \vertices{H}$ and $e \in H$, $\vertices{H^x_e} \cap \vertices{H}_{\geq x} \subseteq e$.
\end{theorem}

Now we are ready to state the main result of this section:

\begin{theorem}
\label{thm:beta-acyc}
Every $\beta$-acyclic hypergraph has a point decomposition of polynomial size and width $1$. Moreover, such a decomposition can be computed in polynomial time. 
\end{theorem}

\begin{proof}
Let $H$ be a $\beta$-acyclic hypergraph with $\beta$-elimination order $<_{\beta}$. The rooted tree $T$ of the point decomposition of $H$ has one vertex $t_x$ for each vertex $x \in \vertices{H}$, plus a special vertex $t_{\bot}$. The root of $T$ is $t_{\bot}$ and its only child is $t_z$, where $z$ is the last vertex in the $\beta$-elimination order of $H$. The remainder of $T$ is then a path, where $t_x$ is the child of $t_y$ if and only if $y$ is the vertex that directly follows $x$ in the $\beta$-elimination order. In particular, for any two vertices $x,y \in \vertices{H}$ we have that $t_x <_T t_y$ if and only if $x <_{\beta} y$.

For any $t_x \in \vertices{T}$, the associated bag $B_{t_x}$ is the set of all points $(y,e) \in \points{H}$ with $x \in e$ and $x \leq_{\beta} y$. The bag of $t_{\bot}$ is an empty set of points. We denote by $\T$ the pair $(T,(B_t)_{t \in \vertices{T}})$.

By definition of a $\beta$-elimination order, for each $t_x \in \vertices{T}$ it holds that $\bcovernum{H|_{B_{t_x}}} = 1$ and the possible sub-bags are of the form $(t_x, e \cap \vertices{H}_{\geq x})$ with $e \in H$. We now describe the directed graph $A$ on the sub-bags of $\T$ that will complete the point decomposition. Given any two sub-bags $s_x=(t_x,S_x)$ and $s_y=(t_y,S_y)$ with $x,y \in \vertices{H}$ and $x <_{\beta} y$, we add an arc from $s_x$ to $s_y$ if one of the following conditions is satisfied:
\begin{itemize}
\item[($\dagger$)] $|S_x| = 1$ and there exist $e,f \in H$ such that $S_x = e \cap \vertices{H}_{\geq x}$, $S_y = f \cap \vertices{H}_{\geq y}$ and $e \in H^y_f$;
\item[($\dagger \dagger$)] $|S_x| > 1$  and there exist $e,f \in H$ such that $S_x = e \cap \vertices{H}_{\geq x}$, $S_y = f \cap \vertices{H}_{\geq y}$, $e \in H^y_f$ and $y \leq_\beta z$, where $z = \min_{<_{\beta}}(S_x \backslash \{x\})$.
\end{itemize}
In addition, if $|S_x| = 1$ we add the arc $((t_x,S_x),(t_{\bot},\emptyset))$. Figure~\ref{fig:ex} shows the construction applied to the $\beta$-acyclic hypergraph $H$ to the left and $\beta$-elimination order 
$x_1 <_{\beta} x_2 <_{\beta} x_0 <_{\beta} x_3$. 

By construction, $A$ is a $\T$-structure. The next claim will be used in conjunction with Lemma~\ref{lem:capelli2} and Theorem~\ref{thm:capelli3} to show that $A$ is decomposable.

\begin{claim}
\label{claim:reachpath}
Let $s_x=(t_x,S_x)$ and $s_y=(t_y,S_y)$ be two sub-bags with $x,y \in \vertices{H}$ and $S_x, S_y \neq \emptyset$, such that there is a directed path in $A$ from $s_x$ to $s_y$. Then, there exist $e,f \in H$ such that $S_x = e \cap \vertices{H}_{\geq x}$, $S_y = f \cap \vertices{H}_{\geq y}$ and $e \in H^y_f$.
\end{claim}

\begin{proof}
We prove the claim by induction on the length of the path. If the path has length $1$ (i.e. $(s_x,s_y)$ is an arc in $A$) then $(s_x,s_y)$ satisfies either  ($\dagger$) or ($\dagger \dagger$) and the claim holds. Now, suppose that the path has length $n > 1$ and that the claim holds for all paths of length $n-1$. Let $z \in \vertices{H}$, $z <_{\beta} y$, be such that $s_z=(t_z,S_z)$ is the predecessor of $s_y$ in the path. (Note that such a vertex $z$ always exists because the special sub-bag $(t_{\bot},\emptyset)$ is a sink in $A$.) By induction, there exist $e_x,f_z \in H$ such that $S_x = e_x \cap \vertices{H}_{\geq x}$, $S_z = f_z \cap \vertices{H}_{\geq z}$ and $e_x \in H^z_{f_z}$.   Also, since $(s_z,s_y)$ is an arc in $A$, it satisfies either ($\dagger$) or ($\dagger \dagger$) and hence there exist $e_z,f_y \in H$ such that $S_z = e_z \cap \vertices{H}_{\geq z}$, $S_y = f_y \cap \vertices{H}_{\geq y}$ and $e_z \in H^y_{f_y}$. In particular, there exists a walk $w_{f_z e_x}$ from $f_z$ to $e_x$ that only contains vertices $\leq_{\beta} z$ and edges $\leq_H f_z$, and a walk $w_{f_y e_z}$ from $f_y$ to $e_z$ that only contains vertices $\leq_{\beta} y$ and edges $\leq_H f_y$.

If $f_z <_H f_y$, then $(w_{f_y e_z},z,w_{f_z e_x})$ is a walk from $f_y$ to $e_x$ that contains only vertices $\leq_{\beta} y$ and edges $\leq_H f_y$. Therefore, we have $e_x \in H^y_{f_y}$ and the claim follows from the edges $e_x,f_y$. If instead we have $f_y <_H f_z$, then by Theorem~\ref{thm:capelli3} we have $f_z \cap \vertices{H}_{\geq y} = e_z \cap \vertices{H}_{\geq y} \subseteq \vertices{H^y_{f_y}} \cap \vertices{H}_{\geq y} \subseteq f_y$. Note that $f_z \cap \vertices{H}_{\geq y} $ cannot be a strict subset of $f_y \cap \vertices{H}_{\geq y}$ because $f_y <_H f_z$. This implies that $f_z \cap \vertices{H}_{\geq y} = f_y \cap \vertices{H}_{\geq y} = S_y$. Finally, we deduce from the inclusion $H^z_{f_z} \subseteq H^y_{f_z}$ that $e_x \in H^y_{f_z}$, and the claim follows from the edges $e_x,f_z$.
\renewcommand{\qedsymbol}{$\blacksquare$}\end{proof}

\begin{claim}
\label{claim:adecomp}
$A$ is decomposable.
\end{claim}

\begin{proof}
We prove the claim by contradiction. Suppose that $A$ is not decomposable, that is, there exist five sub-bags $s, s_x=(t_x,S_x), s_y=(t_y,S_y), s^1_z=(t_z,S^1_z), s^2_z=(t_z,S^2_z)$ with $x,y,z \in \vertices{H}$ and $x \neq y$ such that (i) $(s_x,s)$ and $(s_y,s)$ are arcs in $A$, (ii) neither $(s_x,s_y)$ nor $(s_y,s_x)$ is an arc in $A$, 
and (iii) there are directed paths in $A$ from $s^1_z$ to $s_x$ and from $s^2_z$ to $s_y$. 
By the definition of $A$, we can further assume that none of $S_x,S_y,S^1_z,S^2_z$ is empty.

By Claim~\ref{claim:reachpath}, there exist $f_x,e_z^1,f_y,e_z^2 \in H$ such that $S_x = f_x \cap \vertices{H}_{\geq x}$, $S_y = f_y \cap \vertices{H}_{\geq y}$, $S_z^1 = e_z^1 \cap \vertices{H}_{\geq z}$, $S_z^2 = e_z^2 \cap \vertices{H}_{\geq z}$, $e_z^1 \in H^x_{f_x}$ and $e_z^2 \in H^y_{f_y}$. Without loss of generality we assume $x <_{\beta} y$. 

We distinguish two cases:
\begin{itemize}
\item $f_x \leq_H f_y$. Observe that $z \in e_z^1 \cap e_z^2 \cap \vertices{H}_{\leq x} \subseteq \vertices{H^x_{f_x}} \cap \vertices{H^y_{f_y}} \cap \vertices{H}_{\leq x}$, so by Lemma~\ref{lem:capelli2} we have $H^x_{f_x} \subseteq H^y_{f_y}$. In particular, it holds that $f_x \in H^y_{f_y}$. Since $(s_x,s_y)$ is not an arc in $A$, we can deduce that $|S_x| > 1$; it follows that $s$ is of the form $(t_w,S_w)$ where $w \leq_{\beta} \min_{<_{\beta}}(S_x \backslash \{x\})$. However, the arc $(s_y,s)$ implies that $y <_\beta w$, which means that $(s_x,s_y)$ should have been an arc in $A$, a contradiction.
\item $f_x \geq_H f_y$. Then, we have $z \in \vertices{H^y_{f_x}} \cap \vertices{H^y_{f_y}} \cap \vertices{H}_{\leq y}$, so by Lemma~\ref{lem:capelli2} we have $H^y_{f_y} \subseteq H^y_{f_x}$. By Theorem~\ref{thm:capelli3} it holds that $f_y \cap \vertices{H}_{\geq y} \subseteq f_x$, and in particular $y \in f_x$. Then, since $(s_x,s)$ is an arc in $A$ and $|S_x| = |f_x \cap \vertices{H}_{\geq x}| > 1$ (as it contains both $x$ and $y$), it follows that $s$ is of the form $(t_w,S_w)$ where $w \leq_{\beta} \min_{<_{\beta}}(S_x \backslash \{x\})$. Again, the arc $(s_y,s)$ implies that $y <_\beta w$. Finally, since $y \in S_x \backslash \{x\}$, we have $w \leq_{\beta} \min_{<_{\beta}}(S_x \backslash \{x\}) \leq_{\beta} y <_{\beta} w$, a contradiction.
\end{itemize} 
\renewcommand{\qedsymbol}{$\blacksquare$}\end{proof}

\begin{claim}
\label{claim:tripledec}
The triple $(T,(B_t)_{t \in \vertices{T}}, A)$ is a point decomposition of $H$.
\end{claim}

\begin{proof}
$T$ is a rooted tree, each $B_t$ with $t \in \vertices{T}$ is a set of points, and $A$ is a decomposable $\T$-structure by Claim~\ref{claim:adecomp}. That leaves conditions (i), (ii) and (iii) in the definition of a point decomposition to verify. 

By construction, for any edge $e \in H$, we have that $\points{\{e\}}=\{(v,e):v\in e\} \subseteq {B_{t_x}}$, where $x \in \vertices{H}$ is the smallest vertex in $e$ with respect to $<_{\beta}$. Hence condition (i) holds.

For condition (ii), let $H'$ be a subhypergraph of $H$. Note that $A':=A[H']_\emptyset$ is precisely the subgraph of $A$ induced by 
\[
\{(t_{\bot},\emptyset)\} \cup \{(t_x, \vertices{H'|_{B_{t_x}}}) : x \in \vertices{H}, \vertices{H'|_{B_{t_x}}} \neq \emptyset\}.
\]
because all sub-bags of the form $(t_x, \emptyset)$ with $x \in \vertices{H}$ are isolated sub-bags of non-root vertices of $T$. We show that $A'$ is a realisation of $A$.  Suppose for the sake of contradiction that it is not the case. The only possibility is that $A'$ has two sinks, and one of them is of the form $s_x=(t_x,S_x)$ with $x \in \vertices{H}$ and $S_x \neq \emptyset$. The sub-bag $s_{\bot} = (t_{\bot},\emptyset)$ belongs to $V(A')$, which implies that $|S_x| > 1$ since otherwise $(s_x,s_{\bot})$ would be an arc in $A'$ and hence $s_x$ would not be a sink. Now, let $y=\min_{<_{\beta}}(S_x \backslash \{x\})$, and let $e_x \in H'$ be such that $S_x = e_x \cap \vertices{H}_{\geq x}$. Let $s_y = (t_y,S_y)$ denote the sub-bag $(t_y, \vertices{H'|_{B_{t_y}}})$ and $e_y \in H'$ be such that $S_y = e_y \cap \vertices{H}_{\geq y}$. Note that $S_y$ is not empty because $(y,e_x) \in B_{t_y}$; this implies in particular that $s_y \in \vertices{A'}$. If $e_x \cap \vertices{H}_{\geq y} = e_y \cap \vertices{H}_{\geq y}$ then $(s_x,s_y)$ would be an arc in $A$ because of condition ($\dagger \dagger$) (with $(e,f) = (e_x,e_x)$). Since $s_y \in \vertices{A'}$, this contradicts our hypothesis that $s_x$ is a sink in $A'$. On the other hand, if $e_x \cap \vertices{H}_{\geq y} \neq e_y \cap \vertices{H}_{\geq y}$ then from the facts that $<_{\beta}$ is a $\beta$-elimination order, $e_x \in H'$ and $y \in e_x$, we can further assume that $e_x \cap \vertices{H}_{\geq y} \subset e_y \cap \vertices{H}_{\geq y}$. It follows that $e_x <_H e_y$, and the walk $(e_y,y,e_x)$ implies that $e_x \in H_{e_y}^y$. However, by condition ($\dagger \dagger$) we deduce that $(s_x,s_y)$ is an arc in $A$, a final contradiction. 

For condition (iii), we first prove that for any arc $(s,s')$ of $A'$ where $s=(t_y,S_y)$, $y \in \vertices{H}$ and $s'=(t',S')$ it holds that $S_y \backslash S' = \{y\}$. Observe that $S_y$ always contains $y$, and $S'$ may only contain vertices $z \in \vertices{H}$ with $y <_H z$, so $S_y \backslash S' = \{y\}$ whenever $(s,s')$ satisfies condition ($\dagger$) or if $s' = (t_{\bot},\emptyset)$. If $(s,s')$ satisfies condition ($\dagger \dagger$) instead, then $s' = (t_z,S')$ for some $z \leq_{\beta} \min_{<_{\beta}}(S_y \backslash \{y\})$. Let $e_y, f_z \in H$ be such that $S_y = e_y \cap \vertices{H}_{\geq y}$, $S' = f_z \cap \vertices{H}_{\geq z}$ and $e_y \in H^z_{f_z}$. By Theorem~\ref{thm:capelli3} we have that $S_y \backslash \{y\} = e_y \cap \vertices{H}_{\geq z} \subseteq \vertices{H^z_{f_z}} \cap \vertices{H}_{\geq z} \subseteq f_z$ and hence $S_y \backslash S' = S_y \backslash (f_z \cap \vertices{H}_{\geq z}) = \{y\}$, as claimed.

Now, let $A'$ be a realisation of $A$ and $x \in \cup_{(t,S) \in \vertices{A'}} S$. It follows from the property above that if $t'$ is the parent of $t$ in $T_{A'}$ and $(t,S), (t',S')$ are the sub-bags in $V(A')$, then $x\in S$ and $x\not\in S'$ if and only if $t = t_x$. Since $x$ may only appear in a set $S_y$ for sub-bags of the form $(t_y,S_y)$ with $y \leq_\beta x$, the set 
\[
\{ t \in V(T_{A'}) : \text{$\exists (t,S)\in V(A')$ and $x\in S$}\}
\]
induces a connected subtree of $T_{A'}$, which proves the claim.
\renewcommand{\qedsymbol}{$\blacksquare$}\end{proof}

The point decomposition $(T,(B_t)_{t \in \vertices{T}}, A)$ has polynomial size. Moreover, it can be computed in polynomial time  since a $\beta$-elimination order can be computed efficiently from $H$~\cite{brault-baron15:stacs}. 
Recall that for each $t_x \in \vertices{T}$ it holds that $\bcovernum{H|_{B_{t_x}}} = 1$; it follows that $(T,(B_t)_{t \in \vertices{T}}, A)$ has width $1$. Together with Claim~\ref{claim:tripledec}, these last observations establish Theorem~\ref{thm:beta-acyc}.
\end{proof}

In the case of the hypergraph $H$ of Figure~\ref{fig:ex}, it can be verified that our construction produces a non-flat point decomposition independently of the $\beta$-elimination order we pick for $H$. As we shall see in the next section, this is not coincidence as $\beta$-acyclic hypergraphs cannot be captured by flat point decompositions of \emph{any} constant width. 
The reason is that the latter captures precisely the so-called hypergraphs of constant MIM-width, which are known to be incomparable with $\beta$-acyclic hypergraphs~\cite{brault-baron15:stacs}. 

\section{Flat point-width and MIM-width}
\label{sec:mim}

In this section, we show how our main tractability result from Theorem \ref{thm:alg} also explains the tractability of \problem{Max-CSP}s for classes of hypergraphs of 
bounded MIM-width \cite{Vatshelle12:phd,Saether15:jair}. Before doing so, we need some notation and definitions. 

An \emph{induced matching} in a graph $G$ is a set $M \subseteq E(G)$ such that no two edges
of $M$ share a common vertex and for every edge $e=\{u,v\}\in E(G)\setminus M$, we have $\{u,v\}\not\subseteq \bigcup_{\{u',v'\} \in M}\{u',v'\}$. 
For a graph $G$, we denote by $\mim{G}$ the maximum size of an induced matching in $G$. 
A graph $G$ is \emph{bipartite} if there is a partition $V_1,V_2$ of its vertex set $V(G)$ such that every edge of $G$ has one endpoint in $V_1$ and the other in $V_2$. 
For a graph $G$ and disjoint subsets $V_1,V_2$ of $V(G)$, we define $G[V_1,V_2]$ to be the bipartite graph with vertex set $V_1\cup V_2$ that contains all edges of $G$
with one endpoint in $V_1$ and the other in $V_2$. 

A \emph{branch decomposition} of a graph $G$ is a pair $(T,\delta)$ where $T$ is a binary rooted tree 
and $\delta$ is a bijection from $V(G)$ to the leaves of $T$. 
For $t \in V(T)$, we let $T_t$ denote the subtree of $T$ rooted at $t$ and $V_t$
denote the set $\{\delta^{-1}(\ell): \text{$\ell$ is a leaf of $T_t$}\}$. The \emph{MIM-width} of the branch
decomposition $(T,\delta)$ is the maximum $\mim{G[V_t, V(G)\setminus V_t]}$, taken over all $t\in V(T)$.   
The MIM-width~\cite{Vatshelle12:phd} of $G$, denoted by $\mimw{G}$, is the
minimum MIM-width over all branch decompositions of $G$.

The \emph{incidence graph} of a hypergraph $H$, denoted by $\inc{H}$, is the bipartite graph 
with vertex set $V(H)\cup H$ and edge set $\{\{v,e\}: \text{$v\in V(H)$, $e\in H$ and $v\in e$}\}$. 
We define the MIM-width $\mimw{H}$ of the hypergraph $H$ to be $\mimw{\inc{H}}$. 
It follows from the work of S{\ae}ther, Telle and Vatshelle \cite{Saether15:jair} that \problem{Max-CSPs} are tractable for hypergraphs 
of bounded MIM-width, provided a branch decomposition of bounded MIM-width is given with the input. More formally:

\begin{theorem}[\cite{Saether15:jair}]
\label{theo:mim-stv}
Let $k\geq 1$ be fixed. 
There exists an algorithm which, given as input a \problem{Max-CSP} instance $I$ with hypergraph $H$ and a branch decomposition of $\inc{H}$ of MIM-width at most $k$, 
computes $\opt{I}$ in time polynomial in $\Vert I \Vert$.
\end{theorem}

Let us stress that the results in~\cite{Vatshelle12:phd,Saether15:jair} are given for Max-SAT (and \#SAT). 
However, Theorem~\ref{theo:mim-stv} can be obtained by adapting the algorithm from~\cite{Vatshelle12:phd,Saether15:jair} to Max-CSPs. 
We omit the details as Theorem~\ref{theo:mim-stv} is implied by the results of this section.

The goal of this section is to prove the following:

\begin{theorem}
\label{theo:main-mim}
Let $k\geq 1$ be fixed. 
For every hypergraph $H$ and branch decomposition of $\inc{H}$ of MIM-width $k$, there exists 
a point decomposition of $H$ of polynomial size in $\Vert H \Vert$ and of width at most $2k$. 
Moreover, this point decomposition can be computed in time polynomial in $\Vert H \Vert$. 
\end{theorem}

Note that we obtain Theorem \ref{theo:mim-stv} as a consequence of Theorem \ref{theo:main-mim} and Theorem \ref{thm:alg}. 
In order to prove Theorem \ref{theo:main-mim}, we show that the MIM-width of a hypergraph is equivalent to its flat point-width modulo constant factors. 
This is the main technical result of this section which we state below:

\begin{theorem}
\label{theo:mim-fpw}
For every hypergraph $H$, we have $\mimw{H}\leq 4\cdot \fpw{H}$ and $\fpw{H}\leq 2\cdot \mimw{H}$.  
Moreover, for a fixed $k\geq 1$, a flat point decomposition (of polynomial size) of width at most $2k$ can be computed in time polynomial in $\Vert H \Vert$ from a branch decomposition of $H$ of MIM-width $k$.
\end{theorem}

Note how Theorem \ref{theo:mim-fpw} directly implies Theorem \ref{theo:main-mim}. 
In order to prove Theorem \ref{theo:mim-fpw}, 
we present several notions of width and show that they are equivalent modulo constant factors. 
As an intermediate step, we show a 
characterisation of the MIM-width of a bipartite graph in terms of its line graph. 
This characterisation of MIM-width and the one from Theorem \ref{theo:mim-fpw} may be of independent interest. 

\subsection{A characterisation of the MIM-width of bipartite graphs}

A \emph{tree decomposition} of a graph $G$ is a pair $(T, (B_t)_{t \in V(T)})$, where $T$ is a tree and each bag $B_t$ is a subset of $V(G)$ such that 
\begin{enumerate}
\item[(i)] $V(G)=\bigcup_{t\in V(T)} B_t$, 
\item[(ii)] for each edge $\{u,v\}\in E(G)$, there exists $t \in V(T)$ such that $\{u,v\} \subseteq B_t$, and 
\item[(iii)] for each $v \in V(G)$ the set $\{ t \in V(T) : v \in B_t\}$ induces a connected subtree of $T$. 
\end{enumerate}
For any function $f:2^{V(G)}\to \qplus$, we define the $f$-width of the decomposition $(T, (B_t)_{t \in V(T)})$ to be the maximum $f(B_t)$, taken over all $t\in V(T)$, 
and the $f$-width of the graph $G$ to be the minimum $f$-width over all its tree decompositions.  
For instance, the standard notion of treewidth~\cite{Robertson84:minors3} corresponds to $s$-width, where $s(X)=|X|-1$, for every $X\subseteq V(G)$.

For a graph $G$, we say that a set $U\subseteq V(G)$ is a \emph{distance-$2$ independent set} 
if for every pair of distinct nodes $u,v\in U$, there is no path from $u$ to $v$ in $G$ of length at most $2$, where the length of a path is the number of edges. 
We denote by $\alpha^2(G)$ the maximum size of a distance-$2$ independent set in $G$. 
For $G$, we define the function $\alpha^2_G:2^{V(G)}\to \qplus$ as $\alpha^2_G(X):=\alpha^2(G[X])$, for every $X\subseteq V(G)$. 
(Recall that $G[X]$ denotes the subgraph of $G$ induced by $X$, i.e., $G[X]=(X,\{\{u,v\}\in E(G): u,v\in X\})$.) 
We also consider the function $\text{mon-}\alpha^2_G:2^{V(G)}\to \qplus$ defined by 
$\text{mon-}\alpha^2_G(X):=\min\{\alpha^2_G(Y): X\subseteq Y\subseteq V(G)\}$, for every $X\subseteq V(G)$. 

\begin{observation}
For a graph $G$, we have the following:
\begin{itemize}
\item $\alpha^2_G$ is \emph{subadditive}, i.e., $\alpha^2_G(X\cup Y)\leq \alpha^2_G(X)+\alpha^2_G(Y)$, for all $X,Y\subseteq V(G)$.
\item $\text{mon-}\alpha^2_G(X)\leq \alpha^2_G(X)$, for all $X\subseteq V(G)$. 
\item $\text{mon-}\alpha^2_G$ is \emph{monotone} (unlike $\alpha^2_G$), i.e., 
$\text{mon-}\alpha^2_G(X)\leq \text{mon-}\alpha^2_G(Y)$, if $X\subseteq Y\subseteq V(G)$.
\end{itemize}
\end{observation}

We are particularly interested in the notions of $\alpha^2_G$-width and $\text{mon-}\alpha^2_G$-width for a graph $G$, 
which we denote by $\alpha^2\text{-w}(G)$ and $\text{mon-}\alpha^2\text{-w}(G)$, 
respectively. 
For a graph $G$, we define the \emph{line graph} of $G$, denoted by $L(G)$, to be the graph with vertex set $E(G)$ such that 
$\{e,f\}$ is an edge in $L(G)$, where $e,f\in E(G)$ and $e\neq f$, if $e$ and $f$ share a common vertex. 

\begin{observation}
\label{obs:mimw-line}
Let $G$ be a graph. Every induced matching in $G$ is a distance $2$-independent set in $L(G)$ and vice versa. 
In particular, $\mim{G}=\alpha^2(L(G))$. 
\end{observation}

Below we show that for bipartite graphs, the MIM-width and the $\alpha^2$-w (and also $\text{mon-}\alpha^2$-w) of the line graph are equivalent, 
modulo constant factors. The proof is an adaptation of the classical equivalence
between treewidth and \emph{branchwidth}~\cite{Robertson91:minorsX}. 

\begin{proposition}
\label{prop:mim-to-alpha}
For every graph $G$, we have $\alpha^2\text{-w}(L(G))\leq 2\cdot \mimw{G}$.
\end{proposition}

\begin{proof}
Given a branch decomposition $(T,\delta)$ of $G$ of MIM-width $k$, we define a tree decomposition of $L(G)$ of $\alpha^2_{L(G)}$-width at most $2k$. 
Recall that for a node $t\in V(T)$, we denote by $T_t$ the subtree of $T$ rooted at $t$ 
and by $V_t$ the set $\{\delta^{-1}(\ell): \text{$\ell$ is a leaf of $T_t$}\}$. 
The underlying tree of our sought decomposition is $T$ itself. 
For $t\in V(T)$, we define $C_t$ to be the set of edges of $G$ appearing in the bipartite graph $G[V_t,V(G)\setminus V_t]$. 
Now we define $B_t$ to be $B_t:=C_t$, if $t\in V(T)$ is a leaf of $T$, and $B_t:=C_t\cup (C_{t_1}\cap C_{t_2})$, otherwise, 
where $t_1$ and $t_2$ are the two children of $t$ in $T$. We claim that $(T, (B_t)_{t \in V(T)})$ satisfies the required conditions.  

For condition (i) of tree decompositions, for every $e=\{u,v\}\in E(G)=V(L(G))$, we have $e\in B_{\delta(u)}$. 
For condition (ii), if $\{e,f\}\in E(L(G))$ and $e\cap f=\{u\}$, then we have $\{e,f\}\subseteq B_{\delta(u)}$. 
In order to prove condition (iii), we show the following properties: 

\begin{enumerate}
\item Suppose $e\in E(G)=V(L(G))$ and $t,t',t''\in V(T)$ are distinct nodes such that $t$ is a descendent of $t'$, $t''$ belongs the (unique) path in $T$ from $t$ to $t'$, 
and $e\in C_{t}\cap C_{t'}$. Then $e\in C_{t''}$. 
\item Suppose $e\in E(G)=V(L(G))$ and $t,t',s\in V(T)$ are distinct nodes such that $t$ and $t'$ are incomparable in $T$, $s$ is the least common ancestor of $t$ and $t'$ in $T$, 
and $e\in C_{t}\cap C_{t'}$. Then $e\in C_{s_1}\cap C_{s_2}$, where $s_1$ and $s_2$ are the two children of $s$ in $T$. 
\end{enumerate}

For property 1), suppose $e=\{u,v\}$, and note that by definition of the $C_t$'s, one endpoint of $e$ belongs to $V_t$, say $u$, and the other endpoint $v$ is in $V(G)\setminus V_{t'}$. 
In particular, $u\in V_{t''}$ and $v\in V(G)\setminus V_{t''}$, and hence $e\in C_{t''}$. 
For property 2), let $e=\{u,v\}$ and note again, by definition of the $C_t$'s, that one endpoint of $e$ belongs to $V_t$, say $u$, and the other endpoint $v$ belongs to $V_{t'}$. 
Then if $s_1$ is the ancestor of $t$, we have $u\in V_{s_1}$ and $v\in V(G)\setminus V_{s_1}$, and therefore $e\in C_{s_1}$. Similarly for $s_2$ and $t'$. 

Now for condition (iii), let $e\in E(G)$ and $t,t',t''$ be distinct nodes in $T$ such that $t''$ belongs to the (unique) path in $T$ from $t$ to $t'$ and $e\in B_{t}\cap B_{t'}$. 
We start with the case when $t$ is a descendent of $t'$ (the case when $t'$ is a descendent of $t$ is analogous). 
Assume first that $e\in C_{t'}\subseteq B_{t'}$. 
We obtain that $e\in C_{t''}\subseteq B_{t''}$, by applying property 1) to $t', t''$ and either $t$ (if $e\in C_t$) or a child of $t$ (if $e\in B_t\setminus C_t$). 
Suppose now that $e\in B_{t'}\setminus C_{t'}$. If $t''$ is a child of $t'$, then $e\in C_{t''}\subseteq B_{t''}$ and we are done. 
Otherwise, if $t'_1$ is the child of $t'$ that is ancestor of $t''$, we obtain $e\in C_{t''}\subseteq B_{t''}$ by applying property 1) to $t'_1$, $t''$ and either $t$ or a child of $t$. 

For the case when $t$ and $t'$ are incomparable, 
we let $s\in V(T)$ be the least common ancestor of $t$ and $t'$ in $T$. 
We obtain that $e\in C_{s_1}\cap C_{s_2}\subseteq B_{s}$, where $s_1$ and $s_2$ are the two children of $s$, 
by applying property 2) to $s$, either $t$ or one of its child, and either $t'$ or one of its child (depending on whether $e\in C_{t}$ and $e\in C_{t'}$, respectively). 
If $t''\neq s$, we can apply the previous case and obtain that $e\in B_{t''}$ as required. 

It remains to bound the $\alpha^2_{L(G)}$-width of $(T, (B_t)_{t \in V(T)})$. 
If $t\in V(T)$ is a leaf of $T$, then $\alpha^2_{L(G)}(B_t)=1$. 
Otherwise let $t_1,t_2$ be the two children of $t$ in $T$. 
By Observation \ref{obs:mimw-line}, we have $\alpha^2_{L(G)}(C_t)\leq k$. 
By subadditivity, we have that $\alpha^2_{L(G)}(B_t)\leq \alpha^2_{L(G)}(C_t)+\alpha^2_{L(G)}(C_{t_1}\cap C_{t_2})$. 
Observe that $C_{t_1}\cap C_{t_2}=E(G[V_{t_1},V_{t_2}])$ (in particular, $L(G)[C_{t_1}\cap C_{t_2}]=L(G[V_{t_1},V_{t_2}])$). 
By Observation \ref{obs:mimw-line}, $\alpha^2_{L(G)}(C_{t_1}\cap C_{t_2})= \mim{G[V_{t_1},V_{t_2}]}$, 
and since $G[V_{t_1},V_{t_2}]$ is an induced subgraph of $G[V_{t_1},V(G)\setminus V_{t_1}]$, we have $\mim{G[V_{t_1},V_{t_2}]}\leq \mim{G[V_{t_1},V(G)\setminus V_{t_1}]}\leq k$. 
We obtain that $\alpha^2_{L(G)}(B_t)\leq 2k$ as required. 
\end{proof}

\begin{proposition}
\label{prop:alpha-to-mim}
For every bipartite graph $G$, we have $\mimw{G}\leq 2\cdot \text{mon-}\alpha^2\text{-w}(L(G))$.
\end{proposition}

\begin{proof}
Let $G$ and $(T, (B_t)_{t \in V(T)})$ be a tree decomposition of $L(G)$ of $\text{mon-}\alpha^2_{L(G)}$-width $k$. 
We can assume that $T$ is a binary rooted tree and that there is a bijection $\delta$ from $V(G)$ to the leaves of $T$ 
such that $B_{\delta(v)}=\{\{v, w\}\in E(G): w\in V(G)\}$, for every $v\in V(G)$. 
To see this, we start by rooting $(T, (B_t)_{t \in V(T)})$ arbitrarily. For each $v\in V(G)$, the set $\{\{v, w\}\in E(G): w\in V(G)\}$ is a clique in $L(G)$, 
and hence there exists $t\in V(T)$ such that $\{\{v, w\}\in E(G): w\in V(G)\}\subseteq B_t$ (note that $t$ is not necessarily unique). We add a fresh leaf $\delta(v)$ to $T$ as a child of $t$ and we let 
$B_{\delta(v)}:=\{\{v, w\}\in E(G): w\in V(G)\}$. 
After this, we iteratively remove all leaves of $T$ that are not of the form $\delta(v)$. 
Since $\text{mon-}\alpha^2_{L(G)}(B_{\delta(v)})=1$, for every $v\in V(G)$, the width of the resulting decomposition is at most $k$. 
Finally, if a node $t$ has $\ell$ children $t_1,\dots,t_\ell$ with $\ell>2$, we force $t$ to have only two children $t_1$ and $t'$, 
where $t'$ is a fresh node with $B_{t'}:=B_t$ and with children $t_2,\dots t_\ell$. 
By applying this modification iteratively, we obtain a rooted binary tree as required. 

We claim that $(T,\delta)$ is a branch decomposition of $G$ of MIM-width at most $2k$. 
Fix $t\in V(T)$. We have that $E(G[V_t, \overline{V}_t])\subseteq B_t$, where $\overline{V}_t:=V(G)\setminus V_t$. 
Indeed, for $e=\{u,v\}\in E(G[V_t, \overline{V}_t])$, 
we have $e\in B_{\delta(u)}\cap B_{\delta(v)}$, and by connectivity, $e\in B_t$. 
Let $V_1,V_2$ be independent sets partitioning $V(G)$ (recall that $G$ is bipartite). 
Let $M\subseteq E(G[V_t,\overline{V}_t])$ be a maximum size induced matching in $G[V_t,\overline{V}_t]$. 
Note that $M$ is the disjoint union of $M_1$ and $M_2$, where $M_1=M\cap E(G[V_t\cap V_1,\overline{V}_t\cap V_2])$ and 
$M_2=M\cap E(G[V_t\cap V_2,\overline{V}_t\cap V_1])$. 
Finally, observe that $M_1$ and $M_2$ are distance $2$-independent sets in $L(G)$ as $V_1$ and $V_2$ are independent sets in $G$. 
In particular, for each $i\in \{1,2\}$, $M_i$ is a distance $2$-independent set in $L(G)[Y]$ for every superset $B_t\subseteq Y$.  
This implies that $|M_i|\leq \text{mon-}\alpha^2_{L(G)}(B_t)$, for $i\in \{1,2\}$. Hence $|M|\leq 2k$. 
\end{proof}

By Propositions \ref{prop:mim-to-alpha} and \ref{prop:alpha-to-mim}, 
for every bipartite graph $G$, we have:

\[\frac{1}{2}\cdot \mimw{G}\leq \text{mon-}\alpha^2\text{-w}(L(G)) 
\leq \alpha^2\text{-w}(L(G))\leq 2\cdot \mimw{G}.\]

\begin{remark}
As in the case of treewidth, the widths $\alpha^2\text{-w}$ and $\text{mon-}\alpha^2\text{-w}$ can be related with other notions such as brambles and games. 
For instance, $\alpha^2\text{-w}$ and $\text{mon-}\alpha^2\text{-w}$ can be lower bounded by the (natural adaptation of the) \emph{bramble number}~\cite{Seymour93:graph}. 
Also, $\text{mon-}\alpha^2\text{-w}$ can be characterised in terms of the monotone version of the \emph{cops and robber game}~\cite{Seymour93:graph} (this is the reason why we work explicitly with $\text{mon-}\alpha^2\text{-w}$ in the first place). 
Now the cops are not restricted to play on a set $X$ of size $k$, but on a set $X$ with $\text{mon-}\alpha^2\text{-w}(X)\leq k$. 
The minimum $k$ for which the cops can win the game in a monotone way is precisely 
the $\text{mon-}\alpha^2\text{-w}$ (this follows for instance from~\cite[Theorem 2.2.12 and Remark 2.1.18]{Adler06:phd}). 
Hence these connections could be used to obtain bounds on the mimw of bipartite graphs. 

\end{remark}

\subsection{Proof of Theorem \ref{theo:mim-fpw}}

We now show the equivalence of fpw and mimw. Let us start with a definition. 

\begin{definition}[Simplified point decomposition]
\label{def:spw}
A \emph{simplified} point decomposition of a hypergraph $H$ is a pair $(T, (B_t)_{t \in \vertices{T}})$ where $T$ is a rooted tree, each set $B_t\subseteq \points{H}$ is a set of points of $H$ and 
\begin{enumerate}
\item[(1)] For every edge $e \in H$, there exists $t \in \vertices{T}$ such that $\points{\{e\}}=\{(v,e):v\in e\} \subseteq B_t$.
\item[(2)] For every subhypergraph $H'$ of $H$, and $v\in \vertices{H'}$, the set 
$\{ t \in \vertices{T} : v\in \vertices{H'|_{B_t}} \}$
induces a connected subtree of $T$.
\end{enumerate}
\end{definition}

As before, the width of a simplified point decomposition $(T, (B_t)_{t \in \vertices{T}})$ is $\max_{t \in \vertices{T}}\bcovernum{H|_{B_t}}$, and the simplified point-width of $H$, denoted by $\spw{H}$, is the minimum width over all its simplified point decompositions.

\begin{proposition}
\label{prop:simplified-crz}
For every hypergraph $H$, we have $\fpw{H}=\spw{H}$. 
\end{proposition}

\begin{proof}
We start by showing $\fpw{H}\leq\spw{H}$. 
Let $(T, (B_t)_{t \in \vertices{T}})$ be a simplified point decomposition of $H$ of width $k$.  
We say that two sub-bags $(t,S)$ and $(t',S)$ with $t\neq t'$ are \emph{consistent} if there exists a subhypergraph $H'$ of $H$ such that $S=\vertices{H'|_{B_t}}$ and $S'=\vertices{H'|_{B_{t'}}}$. 
Consider the triple $(T, (B_t)_{t \in \vertices{T}}, A)$, where $((t,S),(t',S'))$ is an arc in $A$ if and only if $t'$ is the parent of $t$ in $T$ and, $(t,S)$ and $(t',S')$  
are consistent. We claim that $(T, (B_t)_{t \in \vertices{T}}, A)$ is a flat point decomposition of $H$, and hence $\fpw{H}\leq k$. 
Let $H'$ be a subhypergraph of $H$ and note that if $t'$ is the parent of $t$ in $T$ then 
there is an arc from $(t,\vertices{H'|_{B_t}})$ to $(t',\vertices{H'|_{B_{t'}}})$ in $A$ as they are consistent. 
Hence $A[H']_\emptyset$ (actually we have $A[H']_\emptyset=A[H']$) is a realisation of $A$.

Now let $A'$ be an arbitrary realisation of $A$. By definition of $A$, we have that the subtree $T_{A'}$ associated with $A'$ is 
actually a subtree of $T$ that contains the root. By contradiction, suppose the connectivity condition fails for some $v\in \bigcup_{(t,S)\in \vertices{A'}} S$. 
Then, there exists a sequence $(t_0,S_0),\dots, (t_n,S_n)$, with $n\geq 2$, such that (i) each $(t_i,S_i)\in \vertices{A'}$, (ii) $t_0,\dots, t_n$ is a path in $T$, and (iii) $v\in S_0\cap S_n$ but $v\notin S_i$, for $0<i<n$. 
We show by induction that for all $i\in\{1,\dots,n\}$, there exists a subhypergraph $H_i$ of $H$ such that $v\in \vertices{H_i|_{B_{t_0}}}$, $v\not\in \vertices{H_i|_{B_{t_i}}}$ and $S_i\subseteq \vertices{H_i|_{B_{t_i}}}$. 
In particular, $v\not\in \vertices{H_n|_{B_{t_n}}}$ and $S_n\subseteq \vertices{H_n|_{B_{t_n}}}$. This is a contradiction since $v\in S_n$.

For the base case, recall that by construction of $A$, $(t_0,S_0)$ is consistent with $(t_1,S_1)$, and similarly, $(t_1,S_1)$ with $(t_2,S_2)$.  
Hence, there are subhypergraphs $H'_0$ and $H'_1$ of $H$ such that $S_0=\vertices{H'_0|_{B_{t_0}}}$, $S_1=\vertices{H'_0|_{B_{t_1}}}=\vertices{H'_1|_{B_{t_1}}}$ and $S_2=\vertices{H'_1|_{B_{t_2}}}$. 
We define $H_1=H_0'\cup H'_1$. Then we have that $S_0\subseteq \vertices{H_1|_{B_{t_0}}}$ and $S_1=\vertices{H_1|_{B_{t_1}}}$. 
In particular, $v\in S_0\subseteq \vertices{H_1|_{B_{t_0}}}$, $v\not\in S_1= \vertices{H_1|_{B_{t_1}}}$ and $S_1\subseteq \vertices{H_1|_{B_{t_1}}}$, as required. 
For the inductive case, suppose we have $H_i$ with the desired properties, for $i\in\{1,\dots,n-1\}$. As $(t_i,S_i)$ and $(t_{i+1},S_{i+1})$ are consistent, there is a subhypergraph $H_i'$ of $H$ such that 
$S_i=\vertices{H'_i|_{B_{t_i}}}$ and $S_{i+1}=\vertices{H'_i|_{B_{t_{i+1}}}}$. We take $H_{i+1}=H_i\cup H_i'$. 
Note that $S_{i+1}\subseteq \vertices{H_{i+1}|_{B_{t_{i+1}}}}$ and $v\in\vertices{H_{i+1}|_{B_{t_0}}}$ (using the inductive hypothesis $v\in \vertices{H_i|_{B_{t_0}}}$). 
Observe that $\vertices{H_{i+1}|_{B_{t_{i}}}}=\vertices{H_{i}|_{B_{t_{i}}}}\cup S_{i}$. Since $v\not\in S_i$ and $v\not\in \vertices{H_{i}|_{B_{t_{i}}}}$ (by inductive hypothesis), we
 derive that $v\not\in \vertices{H_{i+1}|_{B_{t_{i}}}}$. Since $v\in\vertices{H_{i+1}|_{B_{t_0}}}$, it follows that $v\not\in\vertices{H_{i+1}|_{B_{t_{i+1}}}}$; 
 otherwise the connectivity condition (2) for simplified point decompositions would be violated for $H_{i+1}$. Hence $H_{i+1}$ satisfies all the required conditions.

For $\fpw{H}\geq\spw{H}$, let $(T, (B_t)_{t \in \vertices{T}}, A)$ be a flat point decomposition of $H$ of width $k$. 
We claim that $(T, (B_t)_{t \in \vertices{T}})$ is a simplified point decomposition of $H$, and the result follows. 
Let $H'$ be a subhypergraph of $H$. By definition of point decompositions, 
$A[H']_\emptyset$ is a realisation of $A$ and for every $v\in \vertices{H'}$, the set $\{t\in \vertices{T_{A[H']_\emptyset}}: v\in \vertices{H'|_{B_t}}\}$ 
induces a connected subtree of $T_{A[H']_\emptyset}$. 
For every $t\in V(T)\setminus V(T_{A[H']_\emptyset})$, we have $V(H'|_{B_t})=\emptyset$ and then 
$\{t\in \vertices{T_{A[H']_\emptyset}}: v\in \vertices{H'|_{B_t}}\}$ $=$ $\{t\in \vertices{T}: v\in \vertices{H'|_{B_t}}\}$. 
Since $T_{A[H']_\emptyset}$ must be a subtree of $T$, the latter set induces a connected subtree of $T$.  
Hence condition (2) of Definition~\ref{def:spw} (simplified point decompositions) holds. 
\end{proof}

\medskip

Observe how a simplified point decomposition of $H$ encodes tree decompositions for the subhypergraphs of $H$ without the need of a $\T$-structure, 
unlike the case of flat point decompositions. 
Whether arbitrary point decompositions can also be captured by a notion of decomposition that does not use $\T$-structures explicitly is 
an interesting question which we leave for future work.

For a hypergraph $H$, we define the \emph{point graph} of $H$, denoted by $\pointgraph{H}$, as
\[\pointgraph{H}:=(\points{H}, \{\{(v,e), (v',e')\}: \text{$v=v'$ or
$e=e'$}\}).\]
Note that the point graph $\pointgraph{H}$ of $H$ is isomorphic to $L(\inc{H})$. 
There is a known duality between $\beta$-cn and MIM (see e.g.~\cite[Theorem 2.18]{Capelli16:phd}):

\begin{observation}
\label{obs:cn-mim}
For every hypergraph $H$, we have $\bcovernum{H}=\mim{\inc{H}}$.
By Observation \ref{obs:mimw-line}, we have $\bcovernum{H}=\alpha^2({\pointgraph{H}})$. 
\end{observation}

\begin{proposition}
\label{prop:crz-point}
For every hypergraph $H$, we have $\spw{H}\leq \alpha^2\text{-w}(\pointgraph{H})$ and $\alpha^2\text{-w}(\pointgraph{H})\leq 2\cdot\spw{H}$. 
\end{proposition}

\begin{proof}
For $\spw{H}\leq \alpha^2\text{-w}(\pointgraph{H})$, let $(T, (B_t)_{t \in V(T)})$ be a tree decomposition of $\pointgraph{H}$ of $\alpha^2$-width $k$. 
We claim that $(T, (B_t)_{t \in V(T)})$ is a simplified point decomposition of $H$ of width $k$. 
By Observation \ref{obs:cn-mim}, we have $\bcovernum{H|_{B_t}}=\alpha^2(\pointgraph{H|_{B_t}})=\alpha^2({\pointgraph{H}[B_t]})=\alpha^2_{\pointgraph{H}}(B_t)$, for every $t\in V(T)$. 
Hence, the width of $(T, (B_t)_{t \in V(T)})$ is $k$. 
  For condition (1) of Definition~\ref{def:spw}, let $e\in H$ and note that the set $\{(v,e) \in \points{H}: v\in e\}$ forms a clique in $\pointgraph{H}$. 
Hence, there exists $t\in V(T)$ such that $\{(v,e)\in \points{H}: v\in e\}\subseteq B_t$. 
  Towards a contradiction, suppose that condition (2) of Definition~\ref{def:spw}  is violated, i.e., 
there is a subhypergraph $H'$ of $H$, a vertex $v\in V(H')$ and distinct nodes $t_1,t_2,t_3\in V(T)$ such that 
$t_3$ is in the unique path from $t_1$ to $t_2$ in $T$, and $v\in V(H'|_{B_{t_1}})\cap V(H'|_{B_{t_2}})$ but $v\not\in V(H'|_{B_{t_3}})$. 
In particular, there exist edges $e_1,e_2\in H'$ such that $(v,e_1)\in B_{t_1}$, $(v,e_2)\in B_{t_2}$ and $\{(v,e_1), (v,e_2)\}\cap B_{t_3} =\emptyset$. 
Since $\{(v,e_1),(v,e_2)\}$ is an edge in $\pointgraph{H}$, there is a node $t\in V(T)$ such that $\{(v,e_1), (v,e_2)\}\subseteq B_t$. 
Using the connectivity of the tree decomposition $(T, (B_t)_{t \in V(T)})$, we obtain that $\{(v,e_1), (v,e_2)\}\cap B_{t_3} \neq\emptyset$; a contradiction. 

For $\alpha^2\text{-w}(\pointgraph{H})\leq 2\cdot\spw{H}$, let $(T, (B_t)_{t \in V(T)})$ be a simplified point decomposition of $H$ of width $k$. 
We define $T'$ to be the tree obtained from $T$ by subdividing every edge in $E(T)$, i.e., 
replacing every edge $e=\{t_1,t_2\}\in E(T)$ by two edges $\{t_1,t_e\}$ and $\{t_e,t_2\}$, where $t_e$ is a fresh node. 
For $t\in V(T')$, we define $B'_t:=B_t$, if $t\in V(T)$, or $B'_t:=B_{t_1}\cup B_{t_2}$, if $t=t_e$ with $e=\{t_1,t_2\}$.

We claim that $(T', (B'_t)_{t \in V(T')})$ is a tree decomposition of $\pointgraph{H}$.  
First note that, for every point $(v,e)$ in $H$, by condition (1) of simplified point decompositions, there is $t\in V(T)\subseteq V(T')$, 
such that $(v,e)\in B_t=B_{t}'$, and hence condition (i) of tree decompositions holds. 
For condition (ii), suppose $(v,e)$ and $(v',e)$ are points with $v\neq v'$. 
Again by condition (1), we obtain that there is $t\in V(T)\subseteq V(T')$, 
such that $\{(v,e),(v',e)\}\in B_t=B_{t}'$. 
Now suppose that $(v,e)$ and $(v,e')$ are points with $e\neq e'$ and pick $t,t'\in V(T)$ such that 
$(v,e)\in B_t$ and $(v,e')\in B_{t'}$. 
By applying condition (2) of simplified point decompositions to the subhypergraph $H'=\{e,e'\}$, 
we have that $\{(v,e),(v,e')\}\cap B_{s}\neq \emptyset$, for every $s\in V(T)$ in the unique path from $t$ to $t'$ in $T$. 
In particular, there is an edge $\hat{e}=\{s_1,s_2\}$ in this path such that $(v,e)\in B_{s_1}$ and $(v,e')\in B_{s_2}$. 
It follows that $\{(v,e), (v,e')\}\subseteq B_{t_{\hat{e}}}'$, for $t_{\hat{e}}\in V(T')$, and hence condition (ii) holds. 
For a point $(v,e)$ of $H$, condition (iii) follows from applying condition (2) to the subhypergraph $H'=\{e\}$. 
Finally, note that, by Observation \ref{obs:cn-mim} and subadditivity of  $\alpha^2_{\pointgraph{H}}$, 
the $\alpha^2_{\pointgraph{H}}$-width of $(T', (B'_t)_{t \in V(T')})$ is at most $2k$, as required. 
\end{proof}

Theorem \ref{theo:mim-fpw} follows from Propositions~\ref{prop:crz-point},
\ref{prop:simplified-crz}, \ref{prop:mim-to-alpha}, and \ref{prop:alpha-to-mim}. 
Let us stress that given a branch decomposition $(T,\delta)$ of $\inc{H}$ of MIM-width $k\geq 1$, 
we can efficiently compute  a flat point decomposition (of polynomial size) of width at most $2k$. 
By applying the construction in the proof of Proposition~\ref{prop:mim-to-alpha} (and due to Proposition~\ref{prop:crz-point}), 
from $(T,\delta)$ we can efficiently compute a simplified point decomposition for $H$ of width at most $2k$. 
Finally, the construction in the proof of Proposition~\ref{prop:simplified-crz} of a flat point decomposition from the simplified point decomposition of width $2k$, 
in particular, of the $\T$-structure $A$, can be done in polynomial time. 
The main step is given two nodes $t,t'\in V(T)$, where $t'$ is the parent of $t$, and two sub-bags of the form $(t,S_1)$ and $(t',S_2)$, 
to check whether they are consistent. This is equivalent to checking the existence of two subhypergraphs $H_1$ and $H_2$ with $|H_1|\leq 2k$, $|H_2|\leq 2k$, 
such that (i) $S_1=V(H_1|_{B_t})$, $S_2=V(H_2|_{B_{t'}})$, and (ii) $V(H_1|_{B_{t'}})\subseteq S_2$ and $V(H_2|_{B_{t}})\subseteq S_1$. 
This can be checked in polynomial time. 

\section{Conclusions} 
\label{sec:conc}

We have introduced a new width that unifies $\beta$-acyclicity and bounded MIM-width. We
have also identified a novel island of tractability for structurally restricted
Max-CSPs. 
The main open problem is to obtain more general hypergraph properties that lead to tractability, 
and ultimately find the precise boundary of tractability. 
There are many natural hypergraph properties that generalise bounded point-width whose tractability status is unclear (from less to more general): 
bounded \emph{$\beta$-hypertreewidth} ($\beta$-hw)~\cite{GP04}, bounded \emph{$\beta$-fractional hypertreewidth} ($\beta$-fhw), and bounded \emph{$\beta$-submodular width} ($\beta$-subw). 
In particular, we have $\beta\text{-subw} \leq \beta\text{-fhw} \leq \beta\text{-hw} \leq \text{pw}$.  
For precise definitions, see Appendix~\ref{sec:widths}. 

In addition to $\beta$-acyclicity and MIM-width, our notion of point-width also subsumes a width measure called \emph{coverwidth}, 
introduced in~\cite[Section 5.3.2]{Capelli16:phd}. 
In Appendix~\ref{sec:coverwidth}, we show that every class of hypergraphs of bounded coverwidth also has bounded flat point-width, and hence, bounded MIM-width.   
We also show that the converse does not hold, i.e., bounded MIM-width strictly generalises bounded coverwidth.

We have focused on polynomial-time solvability for Max-CSPs. 
Regarding \emph{fixed-parameter tractability} (FPT), 
it is easy to show (cf. Appendix~\ref{sec:marx}) 
that Marx's classification of
CSPs~\cite{Marx13:jacm} implies an FPT classification of \{0,1\}-valued Max-CSPs and the FPT frontier 
is given by the classes with bounded $\beta$-submodular width. 
This classification implies that for a class of unbounded $\beta$-submodular width the $\{0,1\}$-valued, and hence the finite-valued, 
problem \problem{Max-CSP($\mathcal{H},-$)} is not
fixed-parameter (and thus not polynomial-time) tractable. 
Note that a collapse between bounded point-width and bounded $\beta$-submodular width would give us a complete classification of Max-CSPs in terms 
of polynomial time-solvability (and FPT). Hence, a natural research direction is to study the relationship between all these measures (pw, $\beta$-hw, $\beta$-fhw and $\beta$-subw). 
As a related result, which could be interesting in its own right, 
we show (cf. Appendix~\ref{sec:collapse}) that
bounded $\beta$-fractional hypertreewidth collapses to bounded $\beta$-hypertreewidth.

We finish with a few open problems. 
Firstly, we have shown (Theorem~\ref{thm:beta-acyc}) that every $\beta$-acyclic hypergraph has a point
decomposition of polynomial size and width $1$. We do not know whether the
converse is true.  Secondly, as discussed before Example~\ref{exa:point} in Section~\ref{sec:pointdec}, we do
not know whether the problem of checking that a given triple is a point
decomposition admits an efficient algorithm.
Finally, we do not know whether point decompositions of bounded width can be assumed to have polynomial size (hence the dependency on $\Vert P \Vert$ in the statement of Theorem~\ref{thm:alg}, and the importance given to the fact that the decomposition has polynomial size in Theorem~\ref{thm:beta-acyc} and Theorem~\ref{theo:main-mim}). 

\section*{Acknowledgements}

We would like to thank the anonymous referees of both the conference~\cite{crz19:lics}
and this full version of the paper.

\bibliographystyle{elsarticle-num-url-modified-sorted}
\bibliography{max-csp}

\newcommand{\noopsort}[1]{}\newcommand{\Zivny}{\noopsort{ZZ}\v{Z}ivn\'y}
\begin{thebibliography}{10}
\expandafter\ifx\csname url\endcsname\relax
  \def\url#1{\texttt{#1}}\fi
\expandafter\ifx\csname urlprefix\endcsname\relax\def\urlprefix{URL }\fi
\expandafter\ifx\csname href\endcsname\relax
  \def\href#1#2{#2} \def\path#1{#1}\fi

\bibitem{Adler06:phd}
I.~Adler, Width functions for hypertree decompositions, Ph.D. thesis, Albert
  Ludwig University of Freiburg (2006).

\bibitem{BFMMUY81}
C.~Beeri, R.~Fagin, D.~Maier, A.~Mendelzon, J.~Ullman, M.~Yannakakis,
  Properties of acyclic database schemes, in: Proceedings of the 13th Annual
  ACM Symposium on Theory of Computing (STOC'81), 1981, pp. 355--362.
\newblock \href {http://dx.doi.org/10.1145/800076.802489}
  {\path{doi:10.1145/800076.802489}}.

\bibitem{Beeri83:jacm}
C.~Beeri, R.~Fagin, D.~Maier, M.~Yannakakis, On the desirability of acyclic
  database schemes, Journal of the {ACM} 30~(3) (1983) 479--513.
\newblock \href {http://dx.doi.org/10.1145/2402.322389}
  {\path{doi:10.1145/2402.322389}}.

\bibitem{brault-baron15:stacs}
J.~Brault{-}Baron, F.~Capelli, S.~Mengel, Understanding model counting for
  beta-acyclic {CNF}-formulas, in: Proceedings of the 32nd International
  Symposium on Theoretical Aspects of Computer Science (STACS'15), 2015, pp.
  143--156.
\newblock \href {http://dx.doi.org/10.4230/LIPIcs.STACS.2015.143}
  {\path{doi:10.4230/LIPIcs.STACS.2015.143}}.

\bibitem{Bulatov17:focs}
A.~Bulatov, A dichotomy theorem for nonuniform {C}{S}{P}s, in: Proceedings of
  the 58th Annual IEEE Symposium on Foundations of Computer Science (FOCS'17),
  IEEE, 2017, pp. 319--330.
\newblock \href {http://dx.doi.org/10.1109/FOCS.2017.37}
  {\path{doi:10.1109/FOCS.2017.37}}.

\bibitem{Capelli16:phd}
F.~Capelli, {Structural restrictions of CNF-formulas: applications to model
  counting and knowledge compilation}, Ph.D. thesis, Universit\'e Paris Diderot
  (2016).

\bibitem{capelli17:lics}
F.~Capelli, Understanding the complexity of {\#}{SAT} using knowledge
  compilation, in: Proceedings of the 32nd Annual {ACM/IEEE} Symposium on Logic
  in Computer Science (LICS'17), 2017, pp. 1--10.
\newblock \href {http://dx.doi.org/10.1109/LICS.2017.8005121}
  {\path{doi:10.1109/LICS.2017.8005121}}.

\bibitem{CRZ18}
C.~Carbonnel, M.~Romero, S.~\v{Z}ivn\'y, The complexity of general-valued
  {CSP}s seen from the other side, in: Proceedings of the 59th Annual IEEE
  Symposium on Foundations of Computer Science (FOCS'18), IEEE, 2018, pp.
  319--330.
\newblock \href {http://dx.doi.org/10.1109/FOCS.2018.00031}
  {\path{doi:10.1109/FOCS.2018.00031}}.

\bibitem{crz19:lics}
C.~Carbonnel, M.~Romero, S.~\Zivny, Point-width and {M}ax-{C}{S}{P}s, in:
  Proceedings of the 34th Annual {ACM/IEEE} Symposium on Logic in Computer
  Science (LICS'19), 2019, pp. 1--13.
\newblock \href {http://dx.doi.org/10.1109/LICS.2019.8785660}
  {\path{doi:10.1109/LICS.2019.8785660}}.

\bibitem{chandra77:stoc}
A.~K. Chandra, P.~M. Merlin, Optimal implementation of conjunctive queries in
  relational data bases, in: Proceedings of the 9th Annual ACM Symposium on
  Theory of Computing (STOC'77), ACM, 1977, pp. 77--90.
\newblock \href {http://dx.doi.org/10.1145/800105.803397}
  {\path{doi:10.1145/800105.803397}}.

\bibitem{chen10:succinct}
H.~Chen, M.~Grohe, Constraint satisfaction with succinctly specified relations,
  Journal of Computer and System Sciences 76~(8) (2010) 847--860.
\newblock \href {http://dx.doi.org/10.1016/j.jcss.2010.04.003}
  {\path{doi:10.1016/j.jcss.2010.04.003}}.

\bibitem{Dalmau02:width}
V.~Dalmau, P.~G. Kolaitis, M.~Y. Vardi, Constraint {S}atisfaction, {B}ounded
  {T}reewidth, and {F}inite-{V}ariable {L}ogics, in: Proceedings of the 8th
  {I}nternational {C}onference on {P}rinciples and {P}ractice of {C}onstraint
  {P}rogramming ({C}{P}'02), Vol. 2470 of Lecture Notes in Computer Science,
  Springer, 2002, pp. 310--326.
\newblock \href {http://dx.doi.org/10.1007/3-540-46135-3_21}
  {\path{doi:10.1007/3-540-46135-3_21}}.

\bibitem{Diestel10:graph}
R.~Diestel, Graph Theory, {F}ourth Edition, Springer, 2010.

\bibitem{Fagin83:jacm-degrees}
R.~Fagin, Degrees of {A}cyclicity for {H}ypergraphs and {R}elational {D}atabase
  {S}chemes, Journal of the {ACM} 30 (1983) 514--550.
\newblock \href {http://dx.doi.org/10.1145/2402.322390}
  {\path{doi:10.1145/2402.322390}}.

\bibitem{feder98:monotone}
T.~Feder, M.~Y. Vardi, The {C}omputational {S}tructure of {M}onotone {M}onadic
  {S{N}{P}} and {C}onstraint {S}atisfaction: {A} {S}tudy through {D}atalog and
  {G}roup {T}heory, {SIAM} Journal on Computing 28~(1) (1998) 57--104.
\newblock \href {http://dx.doi.org/10.1137/S0097539794266766}
  {\path{doi:10.1137/S0097539794266766}}.

\bibitem{Fischl18:pods}
W.~Fischl, G.~Gottlob, R.~Pichler, General and fractional hypertree
  decompositions: Hard and easy cases, in: Proceedings of the 37th ACM
  SIGMOD-SIGACT-SIGAI Symposium on Principles of Database Systems (PODS'18),
  2018, pp. 17--32.
\newblock \href {http://dx.doi.org/10.1145/3196959.3196962}
  {\path{doi:10.1145/3196959.3196962}}.

\bibitem{frank1975some}
A.~Frank, Some polynomial algorithms for certain graphs and hypergraphs, in:
  Proceedings of the 5th British Combinatorial Conference, 1975, Utilitas
  Mathematica, 1975.

\bibitem{freuder90:complexity}
E.~C. Freuder, Complexity of {K}-{T}ree {S}tructured {C}onstraint
  {S}atisfaction {P}roblems, in: Proceedings of the 8th {N}ational {C}onference
  on {A}rtificial {I}ntelligence ({A}{A}{A}{I}'90), 1990, pp. 4--9.

\bibitem{Gottlob09:icalp}
G.~Gottlob, G.~Greco, F.~Scarcello, {T}ractable {O}ptimization {P}roblems
  through {H}ypergraph-{B}ased {S}tructural {R}estrictions, in: Proceedings of
  the 36th International Colloquium on Automata, Languages and Programming
  (ICALP'09), Part II, Vol. 5556 of Lecture Notes in Computer Science,
  Springer, 2009, pp. 16--30.
\newblock \href {http://dx.doi.org/10.1007/978-3-642-02930-1_2}
  {\path{doi:10.1007/978-3-642-02930-1_2}}.

\bibitem{Gottlob02:jcss-hypertree}
G.~Gottlob, N.~Leone, F.~Scarcello, Hypertree decomposition and tractable
  queries, Journal of Computer and System Sciences 64~(3) (2002) 579--627.
\newblock \href {http://dx.doi.org/10.1006/jcss.2001.1809}
  {\path{doi:10.1006/jcss.2001.1809}}.

\bibitem{GP04}
G.~Gottlob, R.~Pichler, Hypergraphs in model checking: Acyclicity and
  hypertree-width versus clique-width, SIAM J. Comput. 33~(2) (2004) 351--378.
\newblock \href {http://dx.doi.org/10.1137/S0097539701396807}
  {\path{doi:10.1137/S0097539701396807}}.

\bibitem{Grohe07:jacm}
M.~Grohe, The complexity of homomorphism and constraint satisfaction problems
  seen from the other side, Journal of the {ACM} 54~(1) (2007) 1--24.
\newblock \href {http://dx.doi.org/10.1145/1206035.1206036}
  {\path{doi:10.1145/1206035.1206036}}.

\bibitem{Grohe14:talg}
M.~Grohe, D.~Marx, Constraint solving via fractional edge covers, {ACM}
  Transactions on Algorithms 11~(1) (2014) 4:1--4:20.
\newblock \href {http://dx.doi.org/10.1145/2636918}
  {\path{doi:10.1145/2636918}}.

\bibitem{GSS01}
M.~Grohe, T.~Schwentick, L.~Segoufin, When is the evaluation of conjunctive
  queries tractable?, in: Proceedings of the 33th Annual ACM Symposium on
  Theory of Computing (STOC'01), 2001, pp. 657--666.
\newblock \href {http://dx.doi.org/10.1145/380752.380867}
  {\path{doi:10.1145/380752.380867}}.

\bibitem{Hell90:h-coloring}
P.~Hell, J.~Ne\v{s}et\v{r}il, On the {C}omplexity of {${H}$}-coloring, Journal
  of Combinatorial Theory, Series B 48~(1) (1990) 92--110.
\newblock \href {http://dx.doi.org/10.1016/0095-8956(90)90132-J}
  {\path{doi:10.1016/0095-8956(90)90132-J}}.

\bibitem{hell:graphs}
P.~Hell, J.~Ne\v{s}et\v{r}il, Graphs and {H}omomorphisms, Oxford University
  Press, 2004.

\bibitem{Jeavons98:algebraic}
P.~G. Jeavons, On the {A}lgebraic {S}tructure of {C}ombinatorial {P}roblems,
  Theoretical Computer Science 200~(1-2) (1998) 185--204.
\newblock \href {http://dx.doi.org/10.1016/S0304-3975(97)00230-2}
  {\path{doi:10.1016/S0304-3975(97)00230-2}}.

\bibitem{kolaitis98:pods}
P.~G. Kolaitis, M.~Y. Vardi, Conjunctive-query containment and constraint
  satisfaction, in: Proceedings of the 17th {SIGACT-SIGMOD-SIGART} Symposium on
  Principles of Database Systems (PODS'98), 1998, pp. 205--213.
\newblock \href {http://dx.doi.org/10.1145/275487.275511}
  {\path{doi:10.1145/275487.275511}}.

\bibitem{Marx13:jacm}
D.~Marx, Tractable hypergraph properties for constraint satisfaction and
  conjunctive queries, Journal of the ACM 60~(6), article No. 42.
\newblock \href {http://dx.doi.org/10.1145/2535926}
  {\path{doi:10.1145/2535926}}.

\bibitem{montanari74:constraints}
U.~Montanari, Networks of {C}onstraints: {F}undamental properties and
  applications to picture processing, Information Sciences 7 (1974) 95--132.
\newblock \href {http://dx.doi.org/10.1016/0020-0255(74)90008-5}
  {\path{doi:10.1016/0020-0255(74)90008-5}}.

\bibitem{Raghavendra08:everycsp}
P.~Raghavendra, Optimal algorithms and inapproximability results for every
  {C}{S}{P}?, in: Proceedings of the 40th Annual ACM Symposium on Theory of
  Computing (STOC'08), 2008, pp. 245--254.
\newblock \href {http://dx.doi.org/10.1145/1374376.1374414}
  {\path{doi:10.1145/1374376.1374414}}.

\bibitem{Robertson84:minors3}
N.~Robertson, P.~D. Seymour, Graph minors. {III.} {P}lanar tree-width, Journal
  of Combinatorial Theory, Series B 36~(1) (1984) 49--64.
\newblock \href {http://dx.doi.org/10.1016/0095-8956(84)90013-3}
  {\path{doi:10.1016/0095-8956(84)90013-3}}.

\bibitem{Robertson91:minorsX}
N.~Robertson, P.~D. Seymour, {Graph minors. {X}. Obstructions to
  tree-decomposition}, Journal of Combinatorial Theory, Series B 52~(2) (1991)
  153--190.
\newblock \href {http://dx.doi.org/10.1016/0095-8956(91)90061-N}
  {\path{doi:10.1016/0095-8956(91)90061-N}}.

\bibitem{Saether15:jair}
S.~H. S{\ae}ther, J.~A. Telle, M.~Vatshelle, Solving {\#}{SAT} and {MAXSAT} by
  dynamic programming, J. Artif. Intell. Res. 54 (2015) 59--82.
\newblock \href {http://dx.doi.org/10.1613/jair.4831}
  {\path{doi:10.1613/jair.4831}}.

\bibitem{Schaefer78:complexity}
T.~J. Schaefer, The {C}omplexity of {S}atisfiability {P}roblems, in:
  Proceedings of the 10th {A}nnual {A}{C}{M} {S}ymposium on {T}heory of
  {C}omputing ({S}{T}{O}{C}'78), ACM, 1978, pp. 216--226.
\newblock \href {http://dx.doi.org/10.1145/800133.804350}
  {\path{doi:10.1145/800133.804350}}.

\bibitem{Seymour93:graph}
P.~D. Seymour, R.~Thomas, Graph searching and a min-max theorem for tree-width,
  Journal of Combinatorial Theory Series B 58~(1) (1993) 22--33.
\newblock \href {http://dx.doi.org/10.1006/jctb.1993.1027}
  {\path{doi:10.1006/jctb.1993.1027}}.

\bibitem{Tarjan85:dm}
R.~E. Tarjan, Decomposition by clique separators, Discrete Mathematics 55~(2)
  (1985) 221--232.
\newblock \href {http://dx.doi.org/10.1016/0012-365X(85)90051-2}
  {\path{doi:10.1016/0012-365X(85)90051-2}}.

\bibitem{tz16:jacm}
J.~Thapper, S.~\Zivny, The complexity of finite-valued {C}{S}{P}s, Journal of
  the ACM 63~(4), article No. 37.
\newblock \href {http://dx.doi.org/10.1145/2974019}
  {\path{doi:10.1145/2974019}}.

\bibitem{Vatshelle12:phd}
M.~Vatshelle, New width parameters of graphs, Ph.D. thesis, University of
  Bergen (2012).

\bibitem{Yannakakis81:vldb}
M.~Yannakakis, Algorithms for acyclic database schemes, in: Proceedings of the
  7th International Conference on Very Large Data Bases (VLDB'81), {IEEE}
  Computer Society, 1981, pp. 82--94.

\bibitem{Zhuk17:focs}
D.~Zhuk, A proof of {C}{S}{P} dichotomy conjecture, in: Proceedings of the 58th
  Annual IEEE Symposium on Foundations of Computer Science (FOCS'17), IEEE,
  2017, pp. 331--342.
\newblock \href {http://dx.doi.org/10.1109/FOCS.2017.38}
  {\path{doi:10.1109/FOCS.2017.38}}.

\end{thebibliography}

\appendix

\section{Width measures}
\label{sec:widths}

Let $H$ be a hypergraph and $X\subseteq V(H)$. The hypergraph \emph{induced} by $X$, denote by $H[X]$, is defined as 
\[H[X]:=\{e \cap X : e \in H,  e \cap X \neq \emptyset\}.\]
Note that, in general, $H[X]$ is not a subhypergraph of $H$ as defined in Section~\ref{sec:prelim}. 

A \emph{fractional edge cover} of a hypergraph $H$ is a function $\gamma : H \to \qplus$ such that for all $v \in \vertices{H}$, $\sum_{e \in H : v \in e} \gamma(e) \geq 1$, and the fractional edge cover number of $H$, denoted by $\fraccovernum{H}$, is the minimum of $\sum_{e \in H} \gamma(e)$ over all fractional edge covers $\gamma$ of $H$.

A \emph{tree decomposition} of a hypergraph $H$ is a pair $(T, (B_t)_{t \in \vertices{T}})$, where $T$ is a tree and each bag $B_t$ is a subset of $\vertices{H}$ such that (i) for each $e \in H$ there exists $t \in \vertices{T}$ such that $e \subseteq B_t$ and (ii) for each $v \in \vertices{H}$ the set $\{ t \in \vertices{T} : v \in B_t\}$ induces a connected subtree of $T$.

Let $H$ be a hypergraph. For any function $f:2^{V(H)}\to \qplus$, we define the \emph{$f$-width} of a tree decomposition $(T, (B_t)_{t \in \vertices{T}})$ of $H$ as the maximum of $f(B_t)$ taken over all $t \in \vertices{T}$, and the \emph{$f$-width of $H$} as the minimum $f$-width of a tree decomposition of $H$. Given a hypergraph $H$,
\begin{itemize}
\item The \emph{treewidth}~\cite{Robertson84:minors3} of $H$ is its $s$-width, where $s(X) = |X|-1$;
\item The (generalised) \emph{hypertreewidth}~\cite{Gottlob02:jcss-hypertree} of $H$ is its $c$-width, where $c(X) = \covernum{H[X]}$;
\item The \emph{fractional hypertreewidth}~\cite{Grohe14:talg} of $H$ is its $fc$-width, where $fc(X) = \fraccovernum{H[X]}$.
\end{itemize}
The treewidth, hypertreewidth and fractional hypertreewidth of a hypergraph $H$ will be denoted by $\tw{H}$, $\hw{H}$ and $\fhw{H}$, respectively. 
Let us notice that a hypergraph $H$ is $\alpha$-acyclic if and only if $\hw{H}=1$. 

Let $H$ be a hypergraph. If $\mathcal{F}$ is a set of functions from
$2^{\vertices{H}}$ to $\qplus$, we call \emph{$\mathcal{F}$-width of $H$} the
quantity $\sup\{f\text{-width}(H) : f \in \mathcal{F}\}$. A function $f :
2^{\vertices{H}} \to \qplus$ is \emph{edge-dominated} if $f(e) \leq 1$ for
all $e \in H$, and \emph{submodular} if $f(A\cap B)+f(A\cup B)\leq f(A)+f(B)$
for all $A,B\subseteq\vertices{H}$.  The \emph{submodular
width}~\cite{Marx13:jacm} of $H$, denoted by $\subw{H}$, is its
$\mathcal{F}_s$-width, where $\mathcal{F}_s$
is the set of all edge-dominated submodular functions from $2^{\vertices{H}}$ to
$\qplus$ satisfying $f(\emptyset) = 0$.

Given a hypergraph $H$, the \emph{$\beta$-hypertreewidth}~\cite{GP04} (resp. \emph{$\beta$-fractional hypertreewidth}, \emph{$\beta$-submodular width}) of $H$ is the maximum hypertreewidth (resp. fractional hypertreewidth, submodular width) taken over all subhypergraphs of $H$. We denote these quantities by $\bhw{H}$, $\bfhw{H}$ and $\bsw{H}$, respectively. 
Observe that a hypergraph $H$ is $\beta$-acyclic if and only if $\bhw{H}=1$. 

\section{FPT Classification for \{0,1\}-valued Max-CSPs}
\label{sec:marx}

We denote by \problem{$\{0,1\}$-Max-CSP} the restriction of Max-CSP to $\{0,1\}$-valued functions. 
In other words, an instance of \problem{$\{0,1\}$-Max-CSP} is syntactically identical to a CSP instance but the goal is to compute the maximum number of constraints that can be 
simultaneously satisfied. 

We shall consider a parameterised version of
\problem{$\{0,1\}$-Max-CSP($\mathcal{H},-$)} with parameter $|H|$ (we slightly
abuse notation and denote this parameterised problem simply
\problem{$\{0,1\}$-Max-CSP($\mathcal{H},-$)}). 
In particular, \problem{$\{0,1\}$-Max-CSP($\mathcal{H},-$)} is in the class FPT
of fixed-parameter tractable problems if an instance $I$ of
\problem{$\{0,1\}$-Max-CSP($\mathcal{H},-$)} can be solved in time $f(|H|)\cdot |I|^c$, 
where $f$ is any computable function and $c>0$ is a constant. 

\begin{theorem}
  Let $\mathcal{H}$ be a recursively enumerable class of hypergraphs. Then,
  assuming the Exponential Time Hypothesis (ETH),
  \problem{$\{0,1\}$-Max-CSP($\mathcal{H},-$)} is in FPT if and only if
  $\mathcal{H}$ has bounded $\beta$-submodular width. 
\end{theorem}

\begin{proof}
  For the tractability part, suppose $\mathcal{H}$ has bounded $\beta$-submodular width
  and let $I$ be an instance of \problem{$\{0,1\}$-Max-CSP($\mathcal{H},-$)} with
  the underlying hypergraph $H\in\mathcal{H}$. 
  Let $\pi=I_1,\dots, I_r$ be an enumeration of all the sub-instances of $I$ (that
  is, instances obtained from $I$ by removing some constraints) ordered in
  non-increasing order according to the number of constraints (and hence
  according to the number of edges in the underlying hypergraph). To compute the
  optimal value of $I$, it suffices to find the first sub-instance according to
  $\pi$ that has a solution. Since each sub-instance has bounded submodular
  width, the existence of a solution can be checked in FPT by the result
  of~\cite{Marx13:jacm}. Since the number $r$ of all sub-instances is
  bounded in terms of $|H|$, the whole procedure can be done in FPT.

  For the hardness, suppose that $\mathcal{H}$ has unbounded $\beta$-submodular width. 
  Then for each $H\in \mathcal{H}$ we can take a subhypergraph $H'$ such that
the class $\mathcal{H'}:=\{H'\mid H\in \mathcal{H}\}$ has unbounded submodular width. 
By Marx's result~\cite{Marx13:jacm}, assuming ETH, we have that
  \problem{CSP($\mathcal{H'},-$)} is not in FPT. 
It suffices to show that \problem{CSP($\mathcal{H'},-$)} fpt-reduces to
  \problem{$\{0,1\}$-Max-CSP($\mathcal{H},-$)}. Let $I$ be an instance of
  \problem{CSP($\mathcal{H'},-$)} with the underlying hypergraph $H'\in
  \mathcal{H'}$. 
  We start by enumerating $\mathcal{H}$ until we find a hypergraph $H$ that
  contains as a subhypergraph $H'$. By definition of $\mathcal{H'}$, such an $H$
  must exist. Let $J$ be the instance of \problem{$\{0,1\}$-Max-CSP($\mathcal{H},-$)} obtained from $I$ by additionally adding one empty constraint for each edge $e\in H\setminus H'$. We have that $I$ has a solution if and only if the optimal value of $J$ is the number of constraints in $I$. Note that the reduction can be done in FPT time. 
\end{proof}

\section{Collapse of $\beta$-hypertreewidth and $\beta$-fractional hypertreewidth}
\label{sec:collapse}

It follows from the definitions that $\bfhw{H}\leq \bhw{H}$, for every hypergraph $H$. 
In this section we show that $\bhw{H}\leq f(\bfhw{H})$, for a fixed function $f$ (Proposition~\ref{prop:bhw-bfhw}).
The key ingredient of the proof is the following lemma, which we borrow from~\cite{Fischl18:pods}. 
The \emph{VC dimension} of a hypergraph $H$, denoted by $\vc{H}$, is the size of the largest set $X \subseteq \vertices{H}$ such that $H[X] = 2^X$. 
Note that the precise statement of this result as given in~\cite{Fischl18:pods} (in the proof of Theorem 6.1) differs by a factor $\coverfrac{H}$, 
but we believe that this is due to a typographical error on their side.
\begin{lemma}[\cite{Fischl18:pods}]
\label{lem:cover-vc}
For any hypergraph $H$, it holds that
\[
\cover{H} \leq 2^{\vc{H}+2} \cdot \coverfrac{H} \cdot \log (11 \cdot \coverfrac{H})
\]
\end{lemma}

It follows that if a hypergraph $H$ has a fractional edge cover of small weight then it has a small edge cover unless its VC dimension is large. We will combine this fact with a straightforward upper bound on the VC dimension in terms of $\bfhw{H}$.

\begin{lemma}
\label{lem:vcfhw}
For any hypergraph $H$, it holds that
\[
\vc{H} \leq 2 \cdot \bfhw{H}
\]
\end{lemma}

\begin{proof}
Let $X \subseteq \vertices{H}$ be a subset of vertices of size $\vc{H}$ such that $H[X] = 2^X$. 
Let $K_X$ be the complete graph with vertex set $X$. 
Since $K_X$ is a subhypergraph of $H[X]$ and fhw does not increase by taking induced hypergraphs, 
it holds that $\bfhw{H} \geq \bfhw{H[X]} \geq \fhw{K_X}$. Now, let $(T, (B_t)_{t \in \vertices{T}})$ be a tree decomposition of $K_X$ of $fc$-width $\fhw{K_X}$. 
For each $t\in V(T)$, let $\gamma_t$ be a fractional edge cover of $K_X[B_t]$ such that $\sum_{e\in K_X[B_t]}\gamma_t(e)=\coverfrac{K_X[B_t]}$. 
Since $K_X$ is a clique on $X$, there exists $t^*\in V(T)$ such that $B_{t^*}=X$, and hence $K_X[B_{t^*}]=K_X$. It follows that 
\[
\fhw{K_X} \geq \sum_{e \in K_X} \gamma_{t^*}(e) \geq \frac{1}{2} \sum_{v \in X} \sum_{e \in K_X: v \in e} \gamma_{t^*}(e) \geq \frac{1}{2} |X|.
\]
Hence, 
\[
\bfhw{H} \geq \bfhw{H[X]} \geq \fhw{K_X} \geq \frac{1}{2} |X| = \frac{1}{2} \vc{H}.
\]
\end{proof}

\begin{proposition}
\label{prop:bhw-bfhw}
For any hypergraph $H$, it holds that
\[
\bhw{H} \leq 4^{\bfhw{H}+1} \cdot \bfhw{H} \cdot \log (11 \cdot \bfhw{H}).
\]
\end{proposition}

\begin{proof}
Let $H'$ be a subhypergraph of $H$, and $(T, (B_t)_{t \in \vertices{T}})$ be a tree decomposition of $H'$ of $fc$-width at most $\bfhw{H}$. 
By Lemma~\ref{lem:cover-vc} and Lemma~\ref{lem:vcfhw}, for each bag $B_t$ we have
\begin{align*}
\cover{H'[B_t]} &\leq 2^{\vc{H'[B_t]}+2} \cdot \coverfrac{H'[B_t]} \cdot \log (11 \cdot \coverfrac{H'[B_t]})\\
&\leq 2^{\vc{H}+2} \cdot \bfhw{H} \cdot \log (11 \cdot \bfhw{H})\\
&\leq 4^{\bfhw{H}+1} \cdot \bfhw{H} \cdot \log (11 \cdot \bfhw{H})
\end{align*}
and hence the hypertreewidth of $H'$ is at most $4^{\bfhw{H}+1} \cdot \bfhw{H} \cdot \log (11 \cdot \bfhw{H})$. This is true for all choices of subhypergraph $H'$, so the claim follows.
\end{proof}

\begin{corollary}
\label{cor:collapse-bhw-bfhw}
A class of hypergraphs has bounded $\beta$-hypertreewidth if and only if it has bounded $\beta$-fractional hypertreewidth.
\end{corollary}

\section{Coverwidth and MIM-width} 
\label{sec:coverwidth}

In this section, we prove that bounded MIM-width strictly generalises bounded coverwidth. We start with some definitions. 
Let $H$ be a hypergraph and $<$ be an ordering of $V(H)$. 
For $x\in V(H)$, we define $H^x$ to be the set of edges of $H$ that can be reached from $x$ using only vertices $\leq x$. 
More formally, a walk from $x\in V(H)$ to $e\in H$ is a sequence $(x_1,e_1,x_2, e_2,\dots, x_n, e_n)$ with $n\geq 1$ such that $x=x_1$, $e=e_n$, $x_n\in e_n$ and $\{x_i,x_{i+1}\}\subseteq e_i$, for all $1\leq i\leq n-1$. 
Then $e\in H^{x}$ if and only if there is a walk $(x_1,e_1,x_2, e_2,\dots, x_n, e_n)$ from $x$ to $e$ with $x_i\leq x$, for all $1\leq i\leq n$. 
Note that $\{e\in H: x\in e\}\subseteq H^x$. 
We define $H^x[\geq x]:=H^x[V(H^x)_{\geq x}]=\{e \cap V(H^x)_{\geq x}: e \in H^x,  e \cap V(H^x)_{\geq x} \neq \emptyset\}$, 
where $V(H^x)_{\geq x}=\{y\in V(H^x): y\geq x\}$. Observe that $x\in V(H^x[\geq x])$. 
The \emph{coverwidth} of the ordering $<$ is $\max_{x\in V(H)} \bcovernum{H^x[\geq x]}$. 
The \emph{coverwidth} of $H$, denoted by $\cw{H}$, is the minimum coverwidth over all orderings of $V(H)$. 
It was shown in~\cite{Capelli16:phd} that bounded coverwidth implies tractability of Max-CSP:

\begin{theorem}[\protect{\cite{Capelli16:phd}}]
\label{theo:coverwidth-capelli}
Let $k\geq 1$ be fixed. There exists an algorithm which, given as input a Max-CSP instance $I$ with hypergraph $H$ and an ordering of $V(H)$ of coverwidth $\leq k$, 
computes $\opt{I}$ in time polynomial in $\Vert I \Vert$.
\end{theorem}

The main result of this section is the following:

\begin{proposition}
\label{prop:cw-to-spw}
For every hypergraph $H$, we have $\spw{H}\leq \cw{H}$.
\end{proposition}

\begin{proof}
Fix an ordering $<$ of $V(H)$ of coverwidth $\leq k$, where $k:=\cw{H}$. Let $x_{\max}:=\max_{<}(V(H))$. We define $T$ to be the rooted tree with vertex set $\{t_x: x\in V(H)\}$ and root $t_{x_{\max}}$ such that 
$t_y$ is the parent of $t_x$ in $T$ if and only if $|V(H^x[\geq x])|\geq 2$ and $y=\min_{<}(V(H^x[\geq x])\setminus\{x\})$, or $|V(H^x[\geq x])|=1$ and $y=x_{\max}$. 
For $t_x\in V(T)$, we define $B_{t_x}:=\{(y,e): e\in H^x, y\in e, y\geq x\}$. 

We claim that $(T, (B_t)_{t\in V(T)})$ is a simplified point decomposition of $H$ of width $\leq k$. 
To see the bound on the width, note that $H|_{B_{t_x}}=H^x[\geq x]$. For condition (1) of simplified point decompositions, 
given $e\in H$, we have that $\{(y,e): y\in e\}\subseteq B_{t_x}$, where $x=\min_<(\{y: y\in e\})$. 
For condition (2), we need the following claim:

\begin{claim}
\label{claim:H-x}
Suppose that $t_x$ is a descendent of $t_y$ in $T$ and $|V(H^z[\geq z])|\geq 2$, where $t_z$ is the only child of $t_y$ that is ancestor of $t_x$. Then $H^x\subseteq H^y$.
\end{claim}

\begin{proof}
We show the claim by induction. For the base case, assume $t_y$ is the parent of $t_x$, and let $e\in H^x$. 
It follows that there is a walk $\pi_e$ from $x$ to $e$ using vertices $\leq x$. Since $|V(H^x[\geq x])|\geq 2$, we have 
$y=\min_{<}(V(H^x[\geq x])\setminus\{x\})$. In particular, there is $f\in H^x$ with $y\in f$, and hence a walk $\pi_f$ from $x$ to $f$ using vertices $\leq x$. 
We can concatenate $y,\pi_f^{-1}$ and $\pi_e$, where $\pi_f^{-1}$ is the reverse sequence of $\pi_f$, and obtain a walk from $y$ to $e$ using vertices $\leq y$ (since $x<y$). 
Hence, $e\in H^y$. 
Now suppose that $t_x$ is a descendant of $t_z$ and $t_y$ is the parent of $t_z$, where $x<z<y$. Let $e\in H^x$. By induction, $e\in H^z$. Using the same argument as above, we obtain that $e\in H^y$. 
\renewcommand{\qedsymbol}{$\blacksquare$}\end{proof}
 
Let $H'$ be a subhypergraph of $H$. Suppose that $x\in V(H'|_{B_{t_y}})\cap V(H'|_{B_{t_z}})$, and $t_w$ is in the unique path from $t_z$ to $t_y$ in $T$, where $x,y,z,w\in V(H)$. 
Assume first that $t_z$ is a descendant of $t_y$. Since $x\in V(H'|_{B_{t_z}})$, there is a point $(x,e)\in B_{t_z}$ such that $x\in e$ and $e\in H'$. 
By definition of $B_{t_z}$, we have that $e\in H^z$. Since $w\neq x_{\max}$, we can apply Claim~\ref{claim:H-x} and obtain that $e\in H^w$. 
Since $w<y\leq x$, we have $(x,e)\in {B_{t_w}}$. Therefore,  $x\in V(H'|_{B_{t_w}})$. 
Suppose now that $t_z$ and $t_y$ are incomparable in $T$. 
Since $x\in V(H'|_{B_{t_z}})$, there is $e\in H^z$ with $x\in e$ and $e\in H'$. 
Let $t_r$ be the only child of $t_x$ that is ancestor of $t_z$. Since $r\neq x_{\max}$, by Claim~\ref{claim:H-x}, we have that $e\in H^r$. 
As $r<x$, we have that $\{r,x\}\subseteq V(H^r[\geq r])$. We can then apply Claim~\ref{claim:H-x} and deduce that $e\in H^x$. In particular, 
$x\in V(H'|_{B_{t_x}})$. Since $t_z$ and $t_y$ are descendent of $t_x$, we obtain that $x\in V(H'|_{B_{t_w}})$ by applying the previous case. 
Hence condition (2) holds. 
\end{proof}

Together with Theorem~\ref{theo:mim-fpw} and Proposition \ref{prop:simplified-crz}, we obtain that $\mimw{H}\leq 4\cdot \cw{H}$, for every hypergraph $H$. 
In particular, we have:

\begin{corollary}
\label{coro:mim-cw}
Every class of hypergraphs of bounded coverwidth also has bounded MIM-width. 
\end{corollary}

It follows from the proofs of Propositions~\ref{prop:cw-to-spw}, \ref{prop:crz-point} and \ref{prop:alpha-to-mim}, that, given a hypergraph $H$ and an ordering of $V(H)$ of coverwidth $\leq k$, 
we can compute in time polynomial in $\Vert H\Vert$, a branch decomposition of $H$ of MIM-width $\leq 4k$. In particular, we obtain Theorem~\ref{theo:coverwidth-capelli} as a consequence of 
Theorem~\ref{theo:mim-stv}.

Finally, we show that the converse to Corollary~\ref{coro:mim-cw} does not hold:

\begin{proposition}
\label{prop:no-converse}
There exists a class of hypergraphs with bounded MIM-width and unbounded coverwidth. 
\end{proposition}

\begin{proof}
For every $n\geq 1$, we define $H_n$ to be the hypergraph with vertex set $X\cup Y$, where $X=\{x_1,\dots,x_n\}$ and $Y=\{y_1,\dots,y_n\}$ and edges 
$H=\{X\cup \{y\}: y\in Y\}\cup \{Y\cup \{x\}: x\in X\}$. Let $\mathcal{C}:=\{H_n: n\geq 1\}$. We also define $e_x:=Y\cup \{x\}$, for every $x\in X$; and $e_y:=X\cup \{y\}$, for every $y\in Y$. 

We first prove that $\mathcal{C}$ has unbounded coverwidth by showing that $\cw{H_n}\geq n$, for every $n\geq 1$. 
Let $z_1,\dots,z_{2n}$ be any ordering of $V(H_n)$ and assume without loss of generality that $z_1\in X$. 
Observe that $H^{z_1}[\geq z_1]=\{e\in H_n: z_1\in e\}$. Then we have $\{e_{y_1},\dots,e_{y_n}\}\subseteq H^{z_1}[\geq z_1]$. 
Note that $\{e_{y_1},y_1\}$, $\{e_{y_2},y_2\}$, $\dots$, $\{e_{y_n},y_n\}$ is an induced matching of $\inc{H^{z_1}[\geq z_1]}$. 
By Observation~\ref{obs:cn-mim}, we obtain that $\bcovernum{H^{z_1}[\geq z_1]}\geq n$, and hence, the coverwidth of the ordering $z_1,\dots,z_{2n}$ is $\geq n$. 
Since this holds for any ordering, we have that $\cw{H_n}\geq n$. 

Now we show that $\mimw{H_n}\leq 2$, for every $n\geq 1$. We define a branch decomposition $(T,\delta)$ for $\inc{H_n}$ as follows. 
Let $P$ be the rooted path 
\[t_{1,1},t_{1,2}, \dots,t_{n,1}, t_{n,2}, s_{1,1}, s_{1,2},\dots, s_{n,1}, s_{n,2}\]
with root $t_{1,1}$. 
The tree $T$ is obtained from $P$ by adding, for every $1\leq i\leq n$, fresh nodes $t'_{i,1}, t'_{i,2}$, whose parents are $t_{i,1}, t_{i,2}$, respectively; 
and by adding for every $1\leq i\leq n-1$, fresh nodes $s'_{i,1}, s'_{i,2}$, whose parents are $s_{i,1}, s_{i,2}$, respectively, and a fresh node $s'_{n,1}$ with parent $s_{n,1}$. 
For every $1\leq i\leq n$, we let $\delta(t'_{i,1})=x_i$ and $\delta(t'_{i,2})=e_{x_i}$; for every $1\leq i\leq n-1$, we let $\delta(s'_{i,1})=y_i$ and $\delta(s'_{i,2})=e_{y_i}$; 
and we set $\delta(s'_{n,1})=y_n$ and $\delta(s_{n,2})=e_{y_n}$. 

We claim that the MIM-width of $(T,\delta)$ is at most $2$. 
Let $t$ be an internal node (i.e., not a leaf) of $T$. Suppose that $t=t_{i,1}$ for some $1\leq i\leq n$ (the case $t=s_{i,1}$ is analogous). 
Then we have that $\inc{H_n}[V_t,V(\inc{H_n})\setminus V_t]$ is the disjoint union of two complete bipartite graphs: one with partition $(\{x_1,\dots,x_{i-1}\}, \{e_{y_1},\dots,e_{y_n}\})$ and 
the other with partition $(\{e_{x_1},\dots,e_{x_{i-1}}\}, \{y_1,\dots,y_n\})$. In particular, $\mim{\inc{H_n}[V_t,V(\inc{H_n})\setminus V_t]}\leq 2$. 
Now suppose that $t=t_{i,2}$ for some $1\leq i\leq n$ (again, the case $t=s_{i,2}$ is analogous). 
In this case, $\inc{H_n}[V_t,V(\inc{H_n})\setminus V_t]$ is the union of a complete bipartite graph with partition $(\{e_{x_1},\dots,e_{x_{i-1}}\}, \{y_1,\dots,y_n\})$ 
and the graph obtained from the complete bipartite graph with partition $(\{x_1,\dots,x_i\}, \{e_{y_1},\dots,e_{y_n}\})$ by adding the vertex $e_{x_i}$ and the edge $\{x_i,e_{x_i}\}$. 
Hence, $\mim{\inc{H_n}[V_t,V(\inc{H_n})\setminus V_t]}\leq 2$. We conclude that $\mimw{H_n}\leq 2$, and therefore, that $\mathcal{C}$ has bounded MIM-width. 
\end{proof}

\end{document}